\documentclass[11pt,a4paper,oneside]{article}

\usepackage{jheppub}
\usepackage{amsmath,amssymb,amsthm,mathtools,booktabs}
\usepackage{graphicx}
\usepackage{tikz}
\usepackage{svg}
\usepackage[noabbrev]{cleveref}
\usetikzlibrary{arrows.meta,positioning,fit,shapes.geometric}

\DeclareMathOperator{\Cov}{Cov}
\DeclareMathOperator{\diag}{diag}

\newcommand{\dd}{\mathrm{d}}

\newcommand{\KL}{D_{\mathrm{KL}}}
\newcommand{\EEC}{\mathrm{EEC}}

\newcommand{\gev}{\ensuremath{\mathrm{GeV}}}
\newcommand{\pt}{\ensuremath{p_{\mathrm{T}}}}
\newcommand{\avg}[1]{\left\langle #1 \right\rangle}
\newcommand{\conn}[2]{\kappa_{#2}\!\left(#1\right)}
\newcommand{\order}[1]{\mathcal{I}^{(#1)}}

\newtheorem{theorem}{Theorem}
\newtheorem{proposition}[theorem]{Proposition}
\newtheorem{corollary}[theorem]{Corollary}
\theoremstyle{remark}
\newtheorem{remark}{Remark}

\makeatletter
\gdef\@fpheader{}
\makeatother

\begin{document}

\title{From Information Geometry to Jet Substructure: A Triality of Cumulant Tensors, Energy Correlators, and Hypergraphs}
\author[1,3]{Aritra Bal,}
\author[1]{Markus Klute,}
\author[2]{Benedikt Maier,}
\author[3,4]{and Michael Spannowsky}

\affiliation[1]{Institute of Experimental Particle Physics, Karlsruhe Institute of Technology, 76131 Karlsruhe, Germany}
\affiliation[2]{Blackett Laboratory, Imperial College of Science, Technology and Medicine, London, SW7 2AZ, United Kingdom}
\affiliation[3]{Institute for Theoretical Physics, Karlsruhe Institute of Technology, 76131 Karlsruhe, Germany}
\affiliation[4]{Institute for Quantum Materials and Technologies, Karlsruhe Institute of Technology, Karlsruhe 76131, Germany}

\emailAdd{aritra.bal@kit.edu}
\emailAdd{markus.klute@kit.edu}
\emailAdd{benedikt.maier@cern.ch}
\emailAdd{michael.spannowsky@kit.edu}

\abstract{
Pairwise Fisher graphs capture local covariance information, but they cannot distinguish an irreducible multi-observable radiation pattern from a collection of ordinary pairwise correlations. We show that this missing structure is naturally supplied by the higher Fisher tensors. In a finite basis of binned EECs, ECFs, or EFPs, and in the natural exponential-family coordinates generated by that basis, the same local tensor has three equivalent interpretations: it is a coefficient in the local Kullback--Leibler expansion, a connected cumulant of the chosen correlator observables, and a signed weight on a hyperedge linking those observables. This gives an exact Fisher--correlator--hypergraph triality in the local exponential-family embedding.

The triality provides a direct construction of physics-informed hypergraphs from measured or simulated correlator data. Extending the quadratic Fisher matrix to the first non-trivial higher tensor identifies genuinely connected multi-observable radiation patterns, supplies hyperedge weights for higher-order Laplacians and message passing, and gives a principled criterion for compressing observable bases beyond pairwise information. We develop these constructions and spell out why the exact cumulant interpretation is special to natural exponential-family coordinates.

We illustrate the framework in four applications. In a minimal local-KL study, including the cubic Fisher tensor reduces the KL truncation error by about $30\mathrm{x}$ near the reference point and isolates the dominant triplet structure. In a $W\to q\bar q$ versus $t\to bq\bar q$ substructure benchmark, the hypergraph selector improves compressed-basis classification. In a 33-observable basis-design problem, the Fisher hypergraph retains more third-order local response at twelve observables, with $\overline{R}_{(3)}=0.937$ compared with $0.871$ for the pairwise graph. A low-capacity learning benchmark then shows how the same Fisher hyperedges can be used as an interpretable inductive bias for message passing on correlator observables.
}

\maketitle

\section{Introduction}
\label{sec:intro}

Energy correlators provide one of the cleanest ways to describe radiation patterns at colliders. The energy--energy correlation (EEC) measures pairwise energy flow as a function of angular separation, while higher-point energy correlation functions (ECFs) and the energy-flow polynomial (EFP) basis extend the same logic to multi-prong structure \cite{Basham:1978zq,Larkoski:2013eya,Komiske:2017aww,Hofman:2008ar,Moult:2018jzp}. These observables are therefore a natural language for jet substructure, but the way their information is organised matters. A pairwise graph can encode covariances among observables, but it cannot tell whether a triangle of strong pairwise links came from one irreducible three-observable fluctuation or from three unrelated pairwise effects.

Hypergraphs are designed precisely for this kind of multiway structure. Unlike ordinary graphs, whose edges join only pairs of vertices, a hypergraph allows one hyperedge to connect an entire subset of observables, particles, or features at once \cite{Berge:1989,Bretto:2013}. This is why hypergraph methods have become useful in machine learning and collider reconstruction \cite{Zhou:2006,DBLP:journals/corr/abs-1809-09401,Birch-Sykes:2024gij}. What is missing, however, is a principled physics prescription for the hyperedge weights. In collider applications one would like a hyperedge not to be merely an architectural choice, but to correspond to a definite radiation pattern and to a definite contribution to local distinguishability.

Information geometry supplies the relevant hierarchy. The Fisher information matrix is the quadratic coefficient in the local expansion of the Kullback--Leibler divergence, but it is only the rank-two member of a tower of higher Fisher tensors \cite{Rao:1945,Eguchi:1985,Eguchi:1992,Amari:1985,Amari:2000}. The higher-rank tensors encode the first departures from a purely Gaussian, pairwise description of local statistical response. In machine learning, the same Fisher geometry also underlies natural-gradient methods \cite{Amari:1998}; here we use it instead as a bridge between collider observables and higher-order combinatorial structure.

The central observation of this paper is that the bridge becomes exact in a local exponential-family embedding generated by the chosen observable basis. We choose a finite set of binned EECs, ECFs, or EFPs as sufficient statistics and write a local model of the form $p_\theta(x)=h(x)\exp[\theta^a\phi_a(x)-\psi(\theta)]$. This is not a claim that the full space of collider event distributions is globally an exponential family. Rather, it is a local chart around a reference sample, in which the chosen observables generate the deformations being probed. In this natural coordinate system the log-partition function $\psi(\theta)$ is a cumulant-generating function: its derivatives are simultaneously coefficients in the local KL expansion, connected cumulants of the correlator basis, and weights on hyperedges linking the corresponding observables. We call this equivalence the \emph{Fisher--correlator--hypergraph triality}.

The practical consequence is direct. A connected triplet of correlators acquires a statistical meaning because it contributes to local distinguishability, a physical meaning because it represents an irreducible radiation-pattern fluctuation, and an algorithmic meaning because it specifies a hyperedge that should be retained or propagated jointly. The first genuinely new layer begins at rank three, where the cubic Fisher tensor can expose structure that a pairwise Fisher graph necessarily smears into ordinary edges.

Concretely, this paper makes three contributions:
\begin{enumerate}
    \item We derive the higher-order KL expansion in natural exponential-family coordinates and identify the resulting Fisher tensors with connected cumulants of a finite correlator basis.
    \item We use the same tensors to construct signed Fisher hypergraphs, together with the corresponding hypergraph Laplacians and message-passing operators.
    \item We demonstrate the construction in four applications: a minimal local-KL study, a prong-substructure compression benchmark, a higher-order basis-design problem, and a low-capacity physics-informed learning benchmark.
\end{enumerate}
Figure~\ref{fig:triality} summarises this logic and serves as a roadmap for the rest of the paper.

\begin{figure}[t]
    \centering
    \includegraphics[width=0.95\textwidth]{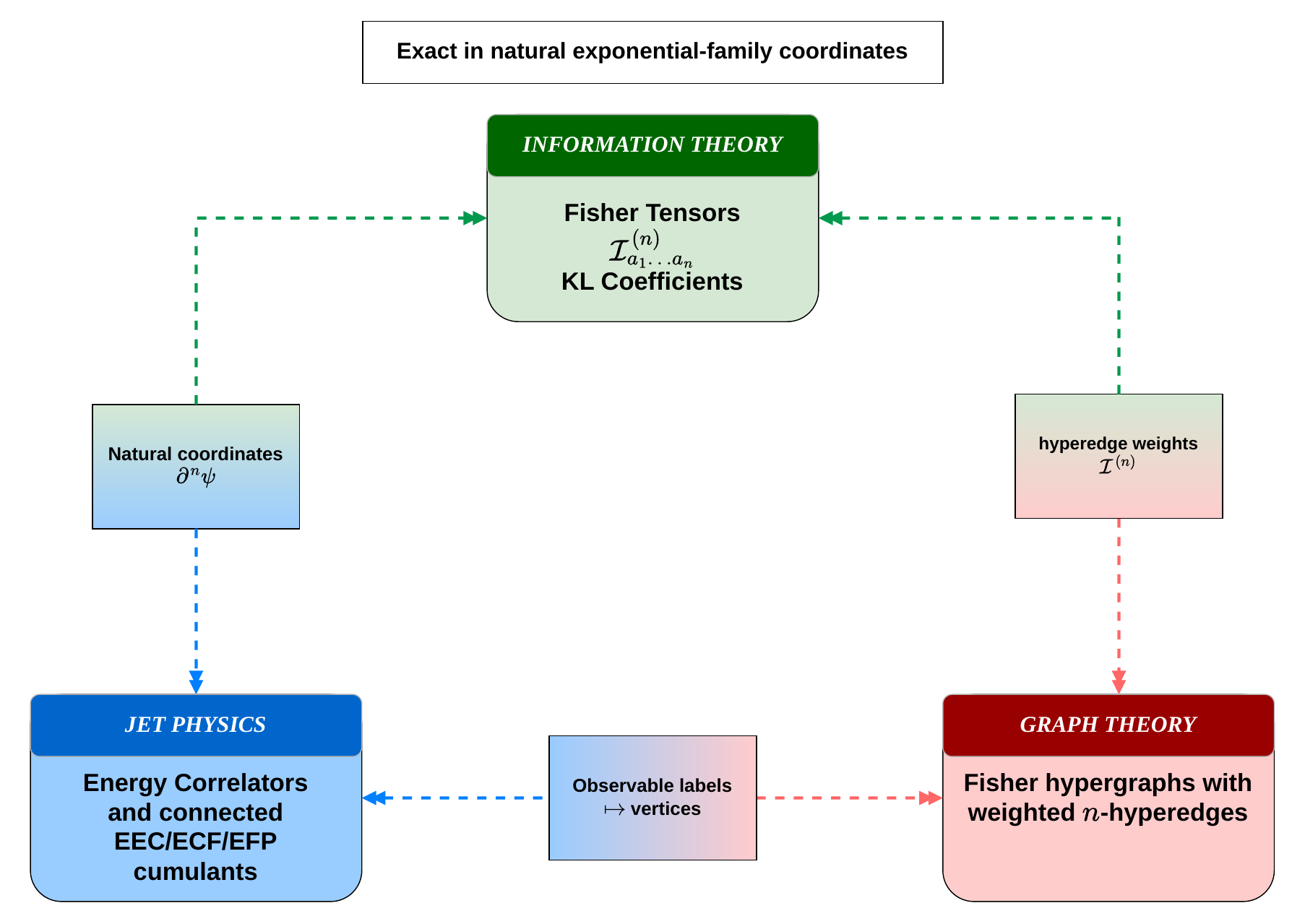}
    \caption{The central triality of the paper. In an exponential-family embedding built from energy correlators, the same tensor is simultaneously a KL coefficient, a connected correlator cumulant, and a hyperedge weight.}
    \label{fig:triality}
\end{figure}

We organise the paper by introducing the three ingredients separately before proving their relation. Section~\ref{sec:fisher} develops the Fisher tensor hierarchy from the KL expansion and states precisely when the cumulant interpretation is exact. Section~\ref{sec:correlators} reviews the EEC, higher-point ECFs, and the EFP basis as the collider observables of interest. Section~\ref{sec:hypergraphs} summarises the weighted hypergraph language, Laplacians, and learning operators that will later inherit physics-motivated weights. Section~\ref{sec:triality} then proves the triality by establishing the three pairwise relations shown in Figure~\ref{fig:triality}. Section~\ref{sec:numerics} gives four concrete applications, beginning with a minimal local-KL study. Appendix~\ref{app:generic} records the generic-coordinate caveat that motivates our use of exponential coordinates.

\section{Fisher Information Tensors}
\label{sec:fisher}

In information geometry, one may expand a smooth statistical divergence around two nearby probability distributions and read off a hierarchy of local tensor fields from its coefficients. For the Kullback--Leibler divergence, the quadratic term yields the Fisher information metric, the cubic term gives the Amari--Chentsov tensor, namely the rank-three member of the Fisher-tensor hierarchy, and the same construction extends to higher symmetric tensors \cite{Rao:1945,Eguchi:1985,Eguchi:1992,Amari:1985,Amari:2000,Ay:2017}. In exponential-family coordinates this hierarchy becomes especially concrete, because the KL coefficients are generated directly by derivatives of the log-partition function. Equivalently, the log-partition function $\psi(\theta)$ is the generating function for the full local hierarchy: its successive derivatives generate the Fisher tensors, just as a cumulant-generating function generates connected moments. We therefore begin by recalling this construction in a form suited to the collider observables used later in the paper.

\subsection{Natural coordinates and cumulant tensors}

Let $x$ denote an event and let $\phi_a(x)$, with $a=1,\ldots,p$, be a chosen observable basis. We consider a model from the exponential family,
\begin{equation}
    p_\theta(x)=h(x)\exp\!\left[\theta^a \phi_a(x)-\psi(\theta)\right],
    \label{eq:exp-family}
\end{equation}
which is one of the standard ways of writing a probability distribution when a chosen set of observables is meant to carry the relevant information. In this form, $h(x)$ is a reference distribution, the $\phi_a(x)$ are the observables or features whose expectation values characterise the event sample, the parameters $\theta^a$ control how strongly the distribution responds to those observables, and $\psi(\theta)$ is the log-normalisation factor that ensures the total probability integrates to one. For readers with a collider background, the main point is simple: once a basis of measured observables has been chosen, the exponential-family form provides a controlled way to describe how the event distribution changes as one tilts the sample towards or away from particular radiation patterns. One may also view the parameters $\theta^a$ as source-like couplings and $\psi(\theta)$ as the associated generating function: differentiating with respect to $\theta$ generates the connected response of the observable basis, and in the present setting those generated objects are precisely the Fisher tensors.

This representation is especially useful here because it makes the link between observables and information geometry explicit. The score components, namely the derivatives of the log-likelihood with respect to the parameters, are
\begin{equation}
    s_a(x;\theta)=\partial_a \log p_\theta(x)=\phi_a(x)-\partial_a \psi(\theta).
    \label{eq:score}
\end{equation}
where it can be shown that $\partial_a \psi(\theta) = \langle \phi_a (x) \rangle_{p_\theta}$, by a single differentiation of the normalisation requirement for Eq.~\eqref{eq:exp-family}. They are therefore just the chosen observables with their mean values subtracted. In other words, the score tells us how unusual a given event looks relative to the average event in the direction of observable $a$. They are centred by construction, $\avg{s_a}_\theta=0$.

The central structural statement is that the full tower of derivatives of $\psi$ is the cumulant tower of the observables, and therefore also the Fisher-tensor hierarchy itself \cite{McCullagh:1987}.

\begin{theorem}[Fisher tensors as cumulants]
\label{thm:cumulants}
For the exponential family in Eq.~\eqref{eq:exp-family},
\begin{equation}
    \partial_{a_1}\cdots \partial_{a_n}\psi(\theta)
    =
    \conn{\phi_{a_1}\cdots \phi_{a_n}}{\theta}
    \equiv
    \order{n}_{a_1\cdots a_n}(\theta),
    \qquad n\ge 1.
    \label{eq:cumulant-derivatives}
\end{equation}
In particular,
\begin{align}
    \order{2}_{ab}(\theta)
    &=\Cov_\theta(\phi_a,\phi_b)
    =\avg{s_a s_b}_\theta,
    \label{eq:fisher-matrix}
    \\
    \order{3}_{abc}(\theta)
    &=\conn{\phi_a \phi_b \phi_c}{\theta}
    =\avg{s_a s_b s_c}_\theta .
    \label{eq:cubic-fisher}
\end{align}
\end{theorem}

\begin{proof}
Normalization implies
\begin{equation}
    \psi(\theta+t)-\psi(\theta)
    =
    \log \avg{\exp(t^a \phi_a)}_\theta .
\end{equation}
Expanding the right-hand side as a multivariate cumulant generating function gives \cite{BarndorffNielsen:1978}
\begin{equation}
    \psi(\theta+t)-\psi(\theta)
    =
    \sum_{n=1}^{\infty}\frac{1}{n!}\,
    \conn{\phi_{a_1}\cdots \phi_{a_n}}{\theta}\,
    t^{a_1}\cdots t^{a_n}.
\end{equation}
Matching coefficients of $t^{a_1}\cdots t^{a_n}$ proves Eq.~\eqref{eq:cumulant-derivatives}. The identities in Eqs.~\eqref{eq:fisher-matrix} and \eqref{eq:cubic-fisher} follow from $s_a=\phi_a-\avg{\phi_a}_\theta$.
\end{proof}

The rank-two tensor $\order{2}$ is the ordinary Fisher information matrix. For $n\ge 3$, $\order{n}$ is not a metric tensor, but it still carries invariant local information about how the model departs from the Gaussian approximation around $\theta$. In natural exponential coordinates, the cubic tensor $\order{3}$ coincides with the Amari--Chentsov tensor, that is, with the rank-three Fisher tensor \cite{Chentsov:1982,Amari:1985,Amari:2000}.

\begin{corollary}[Exact KL expansion]
\label{cor:kl}
For the same exponential family, let $\delta$ denote a small parameter displacement so that $p_{\theta+\delta}$ is a nearby distribution, let $\KL(p_\theta\Vert p_{\theta+\delta})$ be the Kullback--Leibler divergence from $p_\theta$ to $p_{\theta+\delta}$, and let $\order{n}_{a_1\cdots a_n}(\theta)=\partial_{a_1}\cdots\partial_{a_n}\psi(\theta)$ be the rank-$n$ Fisher tensor at $\theta$. Then
\begin{equation}
    \KL\!\left(p_\theta \Vert p_{\theta+\delta}\right)
    =
    \sum_{n=2}^{\infty}\frac{1}{n!}\,
    \order{n}_{a_1\cdots a_n}(\theta)\,
    \delta^{a_1}\cdots \delta^{a_n}.
    \label{eq:exact-kl}
\end{equation}
\end{corollary}

\begin{proof}
Using Eq.~\eqref{eq:exp-family},
\begin{equation}
    \KL\!\left(p_\theta \Vert p_{\theta+\delta}\right)
    =
    \psi(\theta+\delta)-\psi(\theta)-\delta^a \partial_a \psi(\theta).
\end{equation}
Expanding $\psi(\theta+\delta)$ about $\theta$ and canceling the linear term gives Eq.~\eqref{eq:exact-kl}.
\end{proof}

Equation~\eqref{eq:exact-kl} is the precise sense in which the higher Fisher tensors control local distinguishability. The quadratic truncation is the familiar Fisher metric approximation. The cubic term is the first correction that can distinguish a skew local response from a purely Gaussian one.

\begin{remark}[Why we insist on exponential coordinates]
\label{rem:coordinates}
For a generic parametric family, the cubic coefficient of the KL expansion is not simply $\avg{s_a s_b s_c}$. The exact identity mixes third score moments with expectations of second derivatives of the log-likelihood. Appendix~\ref{app:generic} records that formula explicitly. The exponential-family embedding makes the higher-order KL coefficients equal to connected observable cumulants without residue terms.
\end{remark}

\section{Energy Correlators}
\label{sec:correlators}

Energy correlators are infrared- and collinear-safe observables that characterise how the energy carried by an event is distributed over angles, and for that reason they provide a direct probe of QCD radiation patterns. The energy-energy correlator (EEC) is the basic two-point example: it measures pairwise angular energy flow, has a long history as a precision observable in $e^+e^-$ annihilation, and has also become central in conformal collider theory and in modern perturbative studies of event structure \cite{Basham:1978zq,Hofman:2008ar,Moult:2018jzp,Dixon:2019uzg}. Higher-point energy correlation functions (ECFs) extend the same idea to genuinely multiparticle patterns and have become standard tools for resolving jet substructure, especially when multi-prong structure is the discriminating feature \cite{Larkoski:2013eya,Chen:2019bpb,Thaler:2010tr,Larkoski:2017jix}. Energy flow polynomials (EFPs) organise this energy-flow information into a systematic linear basis for infrared- and collinear-safe observables, making it possible to treat a broad class of collider measurements within one common language \cite{Komiske:2017aww}. Related geometric perspectives on the full space of collider events have also been developed \cite{Komiske:2019fks}. In the present work these observables provide the physically motivated basis from which the Fisher tensors and the associated hypergraph weights will be constructed.

\subsection{EECs and ECFs}

For an event with constituents $i=1,\ldots,N$, let
\begin{equation}
    z_i=\frac{E_i}{\sum_j E_j}
    \qquad \text{or} \qquad
    z_i=\frac{p_{T,i}}{\sum_j p_{T,j}},
\end{equation}
depending on whether the observable is defined in $e^+e^-$ or hadron-collider kinematics. The $n$-point energy correlation function with angular exponent $\beta>0$ is
\begin{equation}
    e_n^{(\beta)}(x)=
    \sum_{1\le i_1<\cdots<i_n\le N}
    \left(\prod_{r=1}^{n} z_{i_r}\right)
    \left(\prod_{r<s}\theta_{i_r i_s}^{\beta}\right),
    \label{eq:ecf}
\end{equation}
where $\theta_{ij}$ denotes an opening angle or $\Delta R_{ij}$.

The eventwise EEC is the differential pairwise correlator
\begin{equation}
    \EEC_x(\chi)=
    \sum_{i<j} 2 z_i z_j \,
    \delta\!\big(\cos\chi-\cos\chi_{ij}\big).
    \label{eq:eec}
\end{equation}
The factor of two together with the ordered sum $\sum_{i<j}$ is equivalent to the unordered convention $\sum_{i\neq j} z_i z_j\,\delta(\cos\chi-\cos\chi_{ij})$, and is chosen here so that the kernel relation in Proposition~\ref{prop:eec-ecf} carries the conventional $\tfrac12$ prefactor. This relation is relevant for what follows because it puts the EEC and the ECFs on the same footing. The EEC resolves pairwise energy flow differentially in the angular variable $\chi$, whereas $e_2^{(\beta)}$ is an integrated observable obtained by weighting that same distribution with a kernel. Consequently, a basis built from EEC bins and a basis built from two-point ECFs probes the same underlying pairwise radiation pattern at different levels of resolution. This is the reason one can later reinterpret cumulants of ECF observables as kernel-smeared cumulants of the EEC itself. The simplest instance of this statement is the following identity.

\begin{proposition}[The two-point ECF as an EEC moment]
\label{prop:eec-ecf}
For $e^+e^-$ kinematics,
\begin{equation}
    e_2^{(\beta)}(x)
    =
    \frac12 \int_{-1}^{1}\dd(\cos\chi)\,
    \Big(2\sin\frac{\chi}{2}\Big)^\beta
    \EEC_x(\chi).
    \label{eq:eec-to-ecf}
\end{equation}
\end{proposition}

\begin{proof}
Substituting Eq.~\eqref{eq:eec} into the right-hand side of Eq.~\eqref{eq:eec-to-ecf} gives
\begin{equation}
    \frac12 \sum_{i<j} 2 z_i z_j
    \int_{-1}^{1}\dd(\cos\chi)\,
    \Big(2\sin\frac{\chi}{2}\Big)^\beta
    \delta\!\big(\cos\chi-\cos\chi_{ij}\big).
\end{equation}
The delta function localises the integral at the pairwise angle $\chi_{ij}$, so the expression reduces to
\begin{equation}
    \sum_{i<j} z_i z_j \Big(2\sin \frac{\chi_{ij}}{2}\Big)^\beta,
\end{equation}
which is precisely the spherical definition of $e_2^{(\beta)}$ in Eq.~\eqref{eq:ecf} for $n=2$.
\end{proof}

The proposition makes precise the sense in which the two-point ECF is not an unrelated observable but a coarse-grained version of the EEC. This point is used explicitly in Section~\ref{sec:triality}, where observables that are linear functionals of the EEC inherit an interpretation in terms of kernel-smeared connected EEC cumulants. Higher ECFs and EFPs similarly arise from kernels acting on higher-point energy-flow observables \cite{Komiske:2017aww}. Throughout the rest of the paper we work with a finite correlator basis $\{\phi_a\}_{a=1}^p$ chosen from binned EECs, ECFs, and EFPs.

\section{Hypergraphs}
\label{sec:hypergraphs}

Hypergraphs are combinatorial structures in which the basic relation is allowed to involve more than two objects at once. Formally, a hypergraph $\mathcal{H}=(V,E)$ consists of a vertex set $V$ together with a family $E$ of subsets of $V$, called hyperedges; unlike an ordinary graph edge, which always joins exactly two vertices, a hyperedge may connect any number of vertices \cite{Berge:1989,Bretto:2013}. This extension is important whenever the relevant structure is genuinely multiway. A graph can represent only pairwise links, so a collective three-way dependence must be projected onto several ordinary edges and is thereby stripped of its irreducible character. Hypergraphs have therefore become a natural language for higher-order structures in combinatorics and network science, for spectral clustering and embedding methods, and for modern hypergraph neural networks \cite{Bretto:2013,Zhou:2006,DBLP:journals/corr/abs-1809-09401}. More recently they have also been proposed as collider-physics architectures precisely because they can retain higher-order correlations that pairwise graph constructions smear out \cite{Konar:2023ptv}.

\subsection{Weighted hypergraphs and Laplacians}

Let $V$ be a finite set of vertices. In the present paper these vertices will eventually label observables, but for the moment it is enough to think of them as the nodes of the combinatorial structure. For a fixed order $r\ge 2$, an $r$-uniform weighted hypergraph consists of a collection $E^{(r)}$ of subsets $e\subseteq V$ together with a nonnegative weight $w_e^{(r)}$ assigned to each such subset. The condition $|e|=r$ means that every hyperedge in $E^{(r)}$ connects exactly $r$ vertices. Thus $r=2$ reproduces an ordinary weighted graph, while $r=3$ allows one edge to connect three vertices simultaneously.

It is useful to encode this information in matrices. The incidence matrix $B^{(r)}\in \{0,1\}^{|V|\times |E^{(r)}|}$ records which vertices belong to which hyperedges: its rows are labelled by vertices $a\in V$, its columns are labelled by hyperedges $e\in E^{(r)}$, and its entries are
\begin{equation}
    B^{(r)}_{a e}=
    \begin{cases}
        1, & a\in e, \\
        0, & a\notin e.
    \end{cases}
\end{equation}
The diagonal matrix $W^{(r)}=\diag(w_e^{(r)})$ stores the hyperedge weights. From it one defines the vertex degree
\begin{equation}
    d_a^{(r)}=\sum_{e\in E^{(r)}} B^{(r)}_{ae} w_e^{(r)},
\end{equation}
which measures the total order-$r$ weight incident on vertex $a$. Collecting these degrees into the diagonal matrix $D_v^{(r)}=\diag(d_a^{(r)})$, and writing $\delta_e^{(r)}=|e|=r$ for the degree of a hyperedge itself, one also defines the diagonal hyperedge-degree matrix $D_e^{(r)}=\diag(\delta_e^{(r)})$. With this notation, a standard normalised hypergraph Laplacian is \cite{Zhou:2006}
\begin{equation}
    \mathcal{L}^{(r)}
    =
    I
    -
    \left(D_v^{(r)}\right)^{-1/2}
    B^{(r)} W^{(r)}
    \left(D_e^{(r)}\right)^{-1}
    \left(B^{(r)}\right)^\top
    \left(D_v^{(r)}\right)^{-1/2}.
    \label{eq:laplacian}
\end{equation}
For readers more familiar with ordinary graphs, $\mathcal{L}^{(r)}$ plays the same role as the normalised graph Laplacian: it compares a feature at one vertex to a weighted average of features on vertices that participate in the same hyperedges. The ordinary graph is recovered as the special case $r=2$.

In many applications several hyperedge orders are relevant at once. A multi-order hypergraph therefore retains several values of $r$ and combines their Laplacians as
\begin{equation}
    \mathcal{L}_{\mathrm{multi}}
    =
    \sum_{r=2}^{r_{\max}} \alpha_r \mathcal{L}^{(r)},
    \qquad
    \alpha_r\ge 0.
    \label{eq:multi-laplacian}
\end{equation}
The coefficients $\alpha_r$ control how much each order contributes. For example, $\alpha_2$ weights pairwise information, while $\alpha_3$ weights irreducible three-way information.

This distinction is important for the work at hand because a graph projection can discard exactly the higher-order information that the Fisher tensors are meant to capture. A standard reduction, called a clique expansion, replaces one $r$-hyperedge by the full set of $\binom{r}{2}$ pairwise edges among the same vertices. After that replacement, however, the original meaning of the structure is lost: one can no longer tell whether the resulting triangle of pairwise edges came from a single irreducible three-way correlation or from three separate pairwise correlations. 
Thus, in more concrete terms, the difference is the following. If three observables become informative only when they fluctuate together as a triplet, then one 3-hyperedge records that collective effect directly. An ordinary graph can only draw the three pairwise links between them, and that picture looks identical to a situation in which the three pairs were unrelatedly strong. Whenever the first informative structure appears at order three or higher, the hypergraph is therefore the natural object and the graph is only an approximation.

\subsection{Learning rules}

The same ingredients define higher-order message-passing operators. If $H^{(\ell)}\in \mathbb{R}^{|V|\times d_\ell}$ denotes node embeddings at layer $\ell$, a simple multi-order hypergraph layer is
\begin{equation}
    H^{(\ell+1)}
    =
    \sigma\!\left[
    \sum_{r=2}^{r_{\max}}
    \alpha_r
    \left(D_v^{(r)}\right)^{-1/2}
    B^{(r)} W^{(r)}
    \left(D_e^{(r)}\right)^{-1}
    \left(B^{(r)}\right)^\top
    \left(D_v^{(r)}\right)^{-1/2}
    H^{(\ell)} U_r
    \right],
    \label{eq:hgnn}
\end{equation}
with trainable matrices $U_r$ and nonlinearity $\sigma$ \cite{DBLP:journals/corr/abs-1809-09401}. Equation~\eqref{eq:hgnn} is standard once the hyperedge weights have been specified. In applications with signed weights, one may propagate with $|w_e^{(r)}|$ or $(w_e^{(r)})^2$ while carrying the sign as an auxiliary edge feature, and sparsification can be imposed by thresholding small entries. Related collider work has already argued that hypergraph architectures are natural for IRC-safe multipoint correlator structure \cite{Konar:2023ptv,Birch-Sykes:2024gij}. For collider applications the vertices can label EEC bins, ECFs at different $\beta$, or selected EFP monomials; what the present paper adds is a principled way to choose the hyperedge weights from physics itself.

\section{The Fisher--Correlator--Hypergraph Triality}
\label{sec:triality}

Having introduced the three ingredients separately, we can now connect them pairwise. The next three subsections correspond to the three edges of Figure~\ref{fig:triality}. Taken together, they establish the full triangle.

\subsection{Fisher tensors and energy correlators}
\label{subsec:fisher-correlators}

We now return to the exponential-family model of Eq.~\eqref{eq:exp-family} and choose its sufficient statistics from a finite correlator basis,
\begin{equation}
    \phi_a \in
    \left\{
    \text{binned EECs},
    \ e_n^{(\beta)},
    \ \text{EFPs}
    \right\}.
\end{equation}
The statement below is immediate from Theorem~\ref{thm:cumulants}, but it is the core bridge between information geometry and collider observables, including the higher-point energy correlators studied in \cite{Chen:2019bpb}.

\begin{theorem}[Energy correlator embedding]
\label{thm:ecf-embedding}
Let the sufficient statistics in Eq.~\eqref{eq:exp-family} be chosen from a set of EECs, ECFs, or EFPs. Then
\begin{equation}
    \order{n}_{a_1\cdots a_n}(\theta)
    =
    \conn{\phi_{a_1}\cdots \phi_{a_n}}{\theta},
    \qquad n\ge 2.
    \label{eq:correlator-cumulants}
\end{equation}
If $\phi_a$ is a linear functional of the EEC, as in Eq.~\eqref{eq:eec-to-ecf}, then the same tensor is a kernel-smeared connected EEC cumulant.
\end{theorem}

\begin{proof}
The first statement is Theorem~\ref{thm:cumulants} specialised to a correlator basis. For the second statement, write $\phi_a=\int K_a(\chi)\EEC_x(\chi)\,\dd(\cos\chi)$. Multilinearity of cumulants then gives
\begin{equation}
    \conn{\phi_{a_1}\cdots \phi_{a_n}}{\theta}
    =
    \int \prod_{r=1}^{n}\dd(\cos\chi_r)\,
    K_{a_r}(\chi_r)\,
    \conn{\EEC(\chi_1)\cdots \EEC(\chi_n)}{\theta} .
\end{equation}
\end{proof}

This is the first edge of the triangle. In natural coordinates, the higher Fisher tensors are the connected fluctuations of familiar collider observables.

\subsection{Fisher tensors and hypergraphs}
\label{subsec:fisher-hypergraphs}

The second edge of Figure~\ref{fig:triality} maps the same tensor hierarchy into combinatorial language. Introduce a vertex for every observable label $a$. At fixed order, the full Fisher tensor contains components with repeated indices. For example, at cubic order,
\begin{equation}
    \frac{1}{6}\order{3}_{abc}\delta^a\delta^b\delta^c
    =
    \frac{1}{6}\sum_a \order{3}_{aaa}(\delta^a)^3
    +
    \frac{1}{2}\sum_{a\neq b}\order{3}_{aab}(\delta^a)^2\delta^b
    +
    \sum_{a<b<c}\order{3}_{abc}\delta^a\delta^b\delta^c .
    \label{eq:cubic-squarefree}
\end{equation}
The first two terms are retained in the exact KL expansion, but they are not ordinary set-valued 3-hyperedges: they correspond respectively to single-observable skewness and mixed repeated-observable couplings. Equivalently, the exact tensor defines a weighted multihypergraph whose hyperedges may contain repeated vertices. In the hypergraph constructions below we use the square-free projection, retaining only all-distinct index sets as ordinary hyperedges.

For every order $n\ge 2$ and every unordered all-distinct index set $e=\{a_1,\ldots,a_n\}$, define a signed weight
\begin{equation}
    \widetilde w_e^{(n)}=\order{n}_{a_1\cdots a_n}.
    \label{eq:signed-weight}
\end{equation}
When a diffusion operator or Laplacian needs nonnegative weights, we use
\begin{equation}
    w_e^{(n)}=\left|\widetilde w_e^{(n)}\right|
    \qquad \text{or} \qquad
    w_e^{(n)}=\left(\widetilde w_e^{(n)}\right)^2,
    \label{eq:unsigned-weight}
\end{equation}
and retain the sign as an auxiliary attribute. This is necessary because odd-order cumulants can be negative.
We will refer to the square-free order-$n$ hypergraph obtained from the weights
$\widetilde w_e^{(n)}=\order{n}_{a_1\cdots a_n}$ as the \emph{Fisher $n$-hypergraph}, or simply as the \emph{Fisher hypergraph} when the order is clear from context.

\begin{proposition}[Fisher tensors define Fisher hypergraphs]
\label{prop:fisher-hypergraph}
For each order $n\ge 2$, the all-distinct components of $\order{n}_{a_1\cdots a_n}$ define an $n$-uniform signed hypergraph on the observable index set. The full tensor, including repeated-index components, defines the corresponding weighted multihypergraph and supplies the order-$n$ local divergence coefficients in Eq.~\eqref{eq:exact-kl}.
\end{proposition}

\begin{proof}
The square-free hypergraph is defined by Eq.~\eqref{eq:signed-weight}. Corollary~\ref{cor:kl} identifies the full tensor components as the coefficients of the exact KL expansion. Restricting to all-distinct components gives the set-valued hypergraph used below, while keeping repeated indices gives the exact multihypergraph representation.
\end{proof}

The pairwise Fisher graph is the special case $n=2$. The cubic Fisher tensor is the first genuinely new layer, because it defines 3-hyperedges that cannot be represented faithfully by pairwise edges without losing the distinction between irreducible three-way structure and a dense collection of ordinary correlations.

\subsection{Energy correlators and hypergraphs}
\label{subsec:correlators-hypergraphs}

The remaining edge of the triangle is then a direct physical interpretation of the hypergraph. Because each vertex labels a correlator observable, a hyperedge records an irreducible joint fluctuation among a set of correlators.

\begin{proposition}[Correlator hyperedges]
\label{prop:correlator-hyperedges}
With the square-free construction above, the ordinary $n$-hyperedge $e=\{a_1,\ldots,a_n\}$, with all indices distinct, is present precisely when the connected cumulant $\conn{\phi_{a_1}\cdots \phi_{a_n}}{\theta}$ is nonzero. If each $\phi_a$ is a linear functional of the EEC, then the hyperedge weight is the corresponding kernel-smeared connected EEC cumulant.
\end{proposition}

\begin{proof}
The first statement follows by combining Eq.~\eqref{eq:signed-weight} with Eq.~\eqref{eq:correlator-cumulants}. For the second statement, use the kernel representation $\phi_a=\int K_a(\chi)\EEC_x(\chi)\,\dd(\cos\chi)$ together with multilinearity of cumulants, exactly as in the proof of Theorem~\ref{thm:ecf-embedding}.
\end{proof}

This subsection completes Figure~\ref{fig:triality} at the level of its edges. A hyperedge is not an arbitrary machine-learning primitive: in the present setting it is a connected multipoint fluctuation of collider observables.

\begin{theorem}[Triality]
\label{thm:triality}
For an exponential-family model built from energy correlators, the tensor $\order{n}_{a_1\cdots a_n}$ has three equivalent interpretations:
\begin{enumerate}
    \item it is the coefficient of $\delta^{a_1}\cdots\delta^{a_n}/n!$ in the exact KL expansion of Eq.~\eqref{eq:exact-kl};
    \item it is the connected cumulant of the correlator observables, Eq.~\eqref{eq:correlator-cumulants};
    \item it is the signed weight of the $n$-multihyperedge $\{\!\{a_1,\ldots,a_n\}\!\}$ in the Fisher multihypergraph, with the ordinary Fisher hypergraph obtained by restricting to all-distinct indices.
\end{enumerate}
\end{theorem}

\begin{proof}
Item 1 is Corollary~\ref{cor:kl}. Item 2 is Theorem~\ref{thm:ecf-embedding}. Item 3 is Proposition~\ref{prop:fisher-hypergraph}. All three refer to the same tensor components.
\end{proof}

The theorem illustrates that a measured or simulated correlator cumulant can be interpreted immediately as an information-geometric coefficient and as a hyperedge weight. Conversely, a hypergraph learning rule initialised from $\order{3}$ is initialised from the cubic term of the exact local divergence. At this stage, the three edges displayed already in Figure~\ref{fig:triality} have all been established explicitly.

\section{Applications of the Triality Relation}
\label{sec:numerics}

Once the triality has been established, the same tensor hierarchy can be used in several different ways. In this section we focus on four applications. The first is a minimal toy study, included because it shows in the cleanest possible setting how the cubic Fisher tensor improves the local KL model and simultaneously defines the first non-trivial hypergraph. The second is a prong-substructure compression benchmark, where the triality is used to retain the observables most sensitive to two- versus three-prong jet structure. The third is observable discovery: the hypergraph identifies which correlator collections are genuinely irreducible and therefore worth promoting to a compact analysis basis. The fourth is physics-informed learning: instead of choosing a hypergraph architecture heuristically, one can initialise it directly from the measured or simulated Fisher tensors. All four applications are developed numerically, and all follow from the same map between local distinguishability, correlator cumulants, and hyperedge weights. Before progressing to the actual studies, a quick description of the simulation setup that is used for the generation of Monte-Carlo (MC) samples is provided in the following subsection.

\subsection{Simulation Setup}
\label{subsec:sim_setup}
Simulated event samples for all processes considered in this study are generated using a common
chain of Monte Carlo tools. Hard-scattering matrix elements are computed at leading order (LO)
with \textsc{MadGraph5\_aMC@NLO}~\cite{Alwall:2014hca} using the Standard Model with CKM
quark mixing (\texttt{SM-CKM} model). The \texttt{NNPDF3.1\_nnlo\_as\_0118\_luxqed} parton
distribution function (PDF) set~\cite{Ball:2017nwa,Bertone:2017bme} is adopted throughout.
Parton showering and hadronization are performed with \textsc{Pythia\,8}~\cite{Sjostrand:2014zea},
employing the CMS \texttt{CP5} underlying-event tune~\cite{CMS:2019csb}. The resulting
particle-level events are subsequently passed through \textsc{Delphes\,3}~\cite{deFavereau:2013fsa},
a fast parametric detector simulation, configured with the official CMS detector card distributed
with the package. The specific processes and their corresponding sample sises are described in the
sections that follow.

\subsection{Minimal Study of Local KL and Triplet Structure}
\label{sec:toy-application}

We begin with a minimal class-mixture example because it gives the cleanest low-dimensional view of the full triality. In this four-observable setting one can see directly that the cubic Fisher tensor both improves the local divergence model and isolates the first genuinely irreducible triplet in correlator space.

\subsubsection{Setup}
This study is intentionally simple. We use $100{,}000$ events, retaining up to the $25$ highest-\pt constituents, from a $22\%$ admixture of QCD jets and hadronic top decay jets. For both samples, we apply the following kinematic selections to select sufficiently boosted events in a mass window around the top quark mass:
\begin{itemize}
    \item $p_\mathrm{T,jet} > 300\,$\gev
    \item $m_\mathrm{SD}\in [140,220]\,$\gev
\end{itemize}
The resulting feature distribution is visibly non-Gaussian and therefore has a nonzero cubic Fisher tensor.

The observable basis is
\begin{align}
    \phi_1 &= \EEC_{\mathrm{narrow}}
    =
    \sum_{i<j} z_i z_j\, \mathbf{1}_{0.06 < \Delta R_{ij}\le 0.20},  \label{eq:eec_narrow}
    \\
    \phi_2 &= \EEC_{\mathrm{wide}}
    =
    \sum_{i<j} z_i z_j\, \mathbf{1}_{0.20 < \Delta R_{ij}\le 0.80}, \label{eq:eec_wide}
    \\
    \phi_3 &= e_2^{(1)}, \label{eq:e21}
    \\
    \phi_4 &= e_3^{(1)}. \label{eq:e31}
\end{align}
The first two observables are binned EECs; the last two are ECFs. This basis already mixes the pairwise EEC language with higher-point correlators, making it a minimal setting in which to test the triality.

\subsubsection{Third Order Improves the Local Divergence Model}

We standardise the feature vector and define the physically relevant local direction by the class-mean difference,
\begin{equation}
    v
    \propto
    \mu_{\mathrm{top}}-\mu_{\mathrm{QCD}}
    \label{eq:kl-direction}
\end{equation}
Equation~\eqref{eq:exact-kl} then gives the exact KL series along $v$. Figure~\ref{fig:mc-kl} compares the exact divergence to its quadratic and quadratic-plus-cubic truncations.

\begin{figure}[t]
    \centering
    \includegraphics[width=0.72\textwidth]{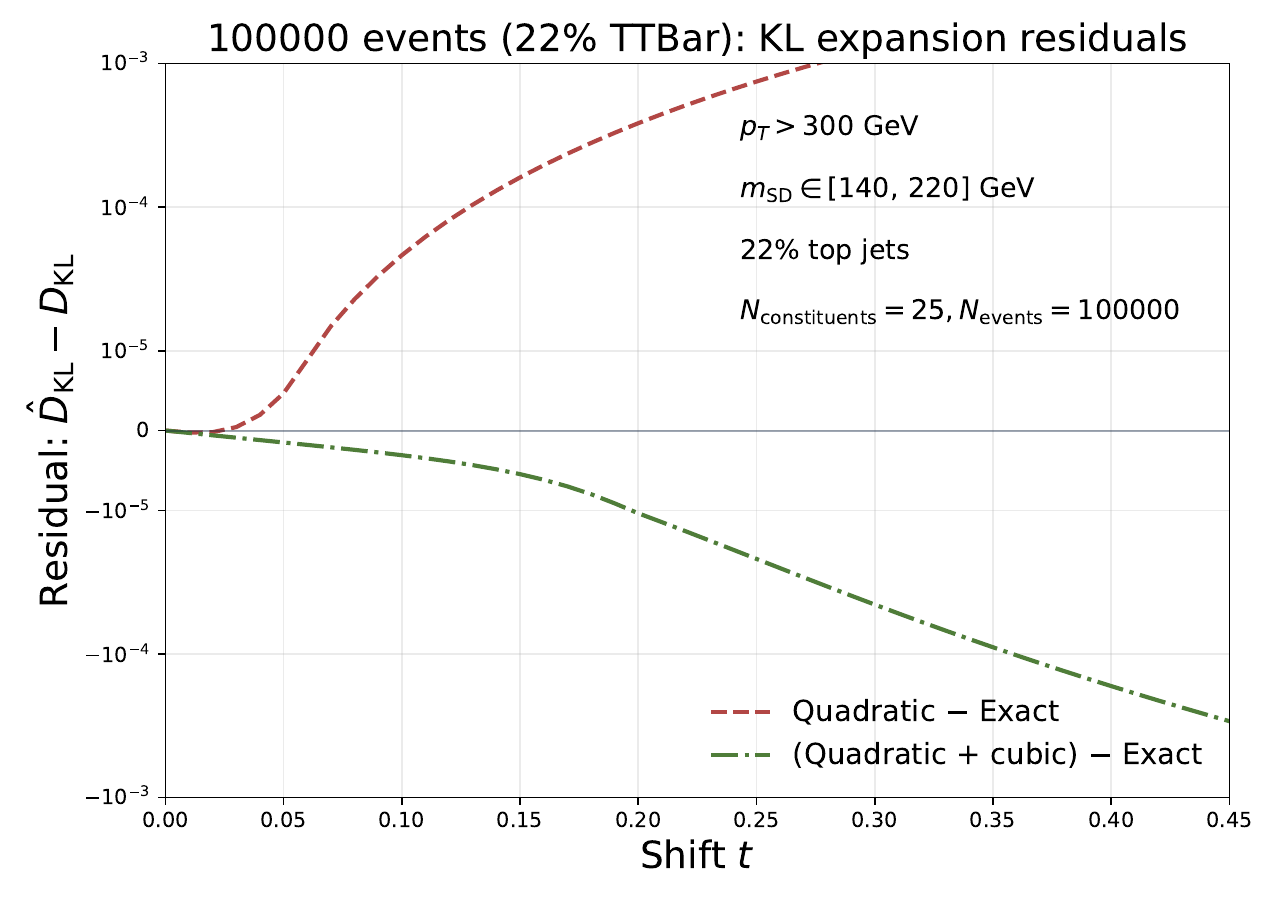}
    \caption{Residuals of the KL divergence along the feature direction in Eq.~\eqref{eq:kl-direction}, for the quadratic Fisher approximation and the approximation including the cubic Fisher tensor, with respect to the exact KL expansion. The cubic term extends the range over which the local information model tracks the true divergence, leading to an approximately $30\mathrm{x}$ reduction in error at the evaluation point $t=0.15$.}
    \label{fig:mc-kl}
\end{figure}

\begin{table}[t]
    \centering
    \small
    \begin{tabular}{ccccc}
        \toprule
        $t$ & exact KL & quadratic & quadratic$+$cubic & error reduction \\
        \midrule
        $0.15$ & $0.023219$ & $0.023380$ & $0.023213$ & $29.5\mathrm{x}$ \\
        $0.30$ & $0.092232$ & $0.093518$ & $0.092187$ & $28.4\mathrm{x}$ \\
        $0.45$ & $0.206216$ & $0.210416$ & $0.205922$ & $14.3\mathrm{x}$ \\
        $0.60$ & $0.364615$ & $0.374073$ & $0.363421$ & $7.9\mathrm{x}$ \\
        \bottomrule
    \end{tabular}
    \caption{Selected points from the exact KL curve. ``Error reduction'' is the ratio of the quadratic truncation error to the quadratic-plus-cubic truncation error. The cubic term is most useful in the genuinely local regime and remains competitive up to moderate shifts.}
    \label{tab:kl}
\end{table}

Table~\ref{tab:kl} makes the same point numerically. The quadratic approximation systematically underestimates the divergence because the local feature distribution is positively skewed along $v$. Adding the cubic Fisher tensor repairs most of that bias for small and moderate shifts. This is the operational meaning of $\order{3}$ in the present setting: it is the first controlled correction to the Fisher-matrix approximation of local distinguishability.

The three legs of the triality are already active in this minimal example. On the Fisher side, $\order{2}$ and $\order{3}$ are the coefficients of the local KL expansion, and Fig.~\ref{fig:mc-kl} shows that the cubic term has direct operational meaning. On the correlator side, those same coefficients are estimated from the EEC/ECF basis as connected moments, so the skewed response along $v$ is tied to a concrete radiation-pattern fluctuation. On the hypergraph side, the same $\order{3}_{abc}$ will become the weight of a 3-hyperedge, allowing the dominant triplet to be read off combinatorially rather than only statistically.

\subsubsection{The Hypergraph Sees Structure the Graph Smears Out}

The pairwise Fisher graph is built from the standardised covariance matrix, while the order-three Fisher hypergraph is built from the distinct cubic weights $\order{3}_{abc}$. The left panel of Figure~\ref{fig:centralities} shows that with standardised features, the pairwise Fisher graph is saturated at second order: all large off-diagonal correlations sit close to $\pm 1$, so pairwise rankings cannot, in this simple study, tell us which observables carry genuinely independent higher-order information. To isolate the third-order content, we move to a whitened basis,
\begin{equation}
z_a = (I^{-1/2})_{ab}\, \big(\phi_b - \avg{\phi_b}_\theta\big),
\end{equation}
where $I_{ab}$ is the second-order Fisher (covariance) matrix. By construction this diagonalises the second-order Fisher tensor, $\tilde I_{ab} = \delta_{ab}$, so that all pairwise correlations among the new features vanish identically; any remaining structure in the third-order tensor $\tilde I_{abc}$ must therefore be irreducible three-point information, decoupled from second-order content.

The whitened features are linear combinations of the original observables, but as shown in Table~\ref{tab:whitening_loadings} the rotated features remain dominated by a single original feature in most cases: $F_0$ by the narrow-angle EEC, $F_1$ by the wide-angle EEC, and $F_3$ by $e_3^{(1)}$. The exception is $F_2$, whose dominant contribution comes from the higher-order $e_3^{(1)}$ rather than from $e_2^{(1)}$, though these contributions are almost equal. Whitening forces all pairwise centralities to be trivially equal and isolates the genuinely third-order centralities, shown in the right panel of Figure~\ref{fig:centralities}. The feature whose dominant contribution comes from $e_3^{(1)}$ is promoted to the most central node of the hypergraph, in contrast to the pairwise graph where $e_2^{(1)}$ ranks highest. Physically, the hypergraph correctly recognises that $e_3^{(1)}$ participates in the strongest irreducible three-prong-sensitive triplet. Figure~\ref{fig:hyperedge_weights} confirms this: the strongest irreducible triplet hyperedge ties together the wide-angle EEC (dominant contributor to $F_1$) with the second- and third-order ECFs (the latter dominantly contributing to $F_2$ and $F_3$).

\begin{table}[ht]
\centering
\begin{tabular}{lcccc}
\toprule
 & $F_0$ & $F_1$ & $F_2$ & $F_3$ \\
\midrule
$\mathrm{EEC_{narrow}}$   & $\mathbf{\phantom{-}0.8547}$ & $\phantom{-}0.1147$ & $\phantom{-}0.0093$ & $-0.0008$ \\
$\mathrm{EEC_{wide}}$     & $\phantom{-}0.1049$ & $\mathbf{\phantom{-}0.6708}$ & $-0.0637$           & $\phantom{-}0.0013$ \\
$e_2^{(1)}$  & $\phantom{-}0.0252$ & $-0.1878$           & $\phantom{-}0.4100$ & $-0.0719$ \\
$e_3^{(1)}$  & $-0.0153$           & $\phantom{-}0.0266$ & $\mathbf{-0.5169}$           & $\mathbf{\phantom{-}0.9260}$ \\
\bottomrule
\end{tabular}
\caption{Contributions to the whitened features $F_i$ in terms of the original observables from \cref{eq:eec_narrow,eq:eec_wide,eq:e21,eq:e31}, with the contributions normalised such that their absolute values sum to $1$.}
\label{tab:whitening_loadings}
\end{table}

\begin{figure}[htbp]
\centering
\includegraphics[width=0.48\textwidth]{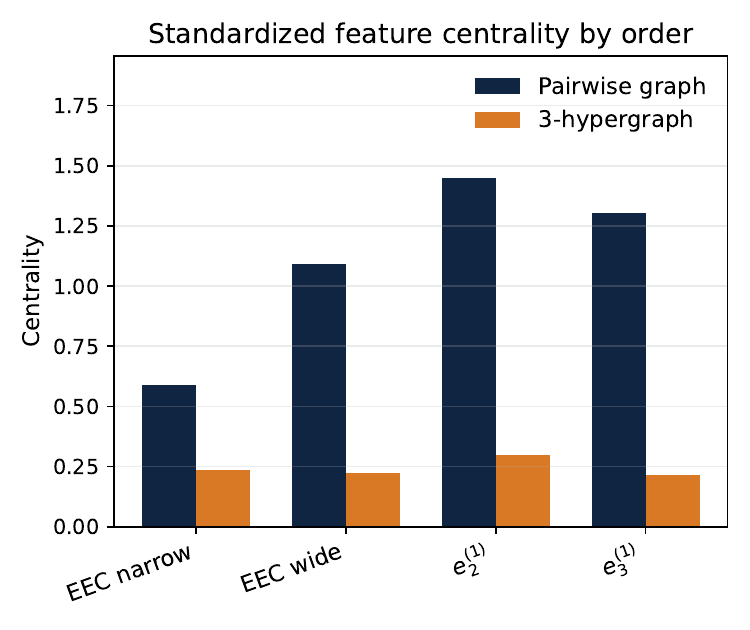}
\hfill
\includegraphics[width=0.48\textwidth]{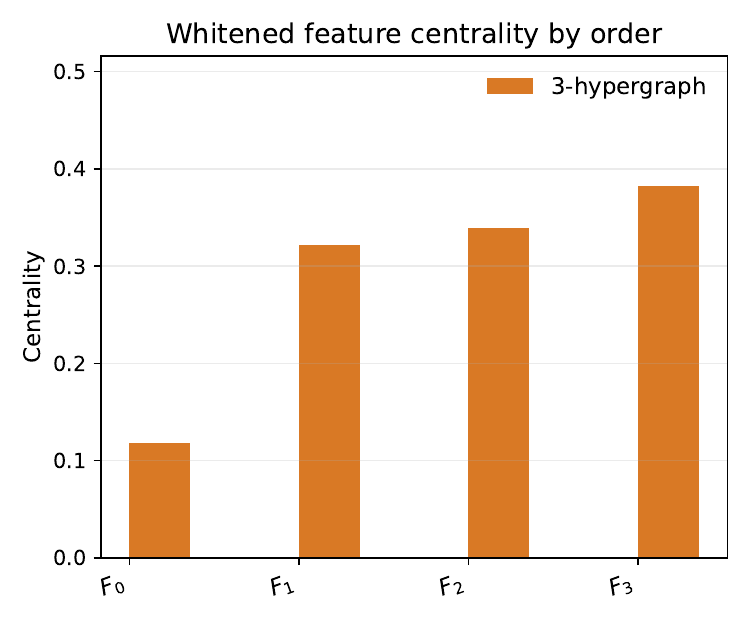}
\caption{Node centralities before and after removing pairwise structure. Left: in the standardised observable basis, the pairwise Fisher graph is saturated, while the orange bars show the corresponding third-order Fisher-hypergraph centralities. Right: in the whitened basis, the second-order Fisher matrix is the identity, so the pairwise graph has no non-trivial off-diagonal edges and the blue centrality bars vanish; the orange bars are therefore the remaining irreducible third-order hypergraph centralities. The scales in both plots are chosen to highlight the relative values.}
\label{fig:centralities}
\end{figure}

\begin{figure}[htbp]
\centering
\includegraphics[width=0.7\textwidth]{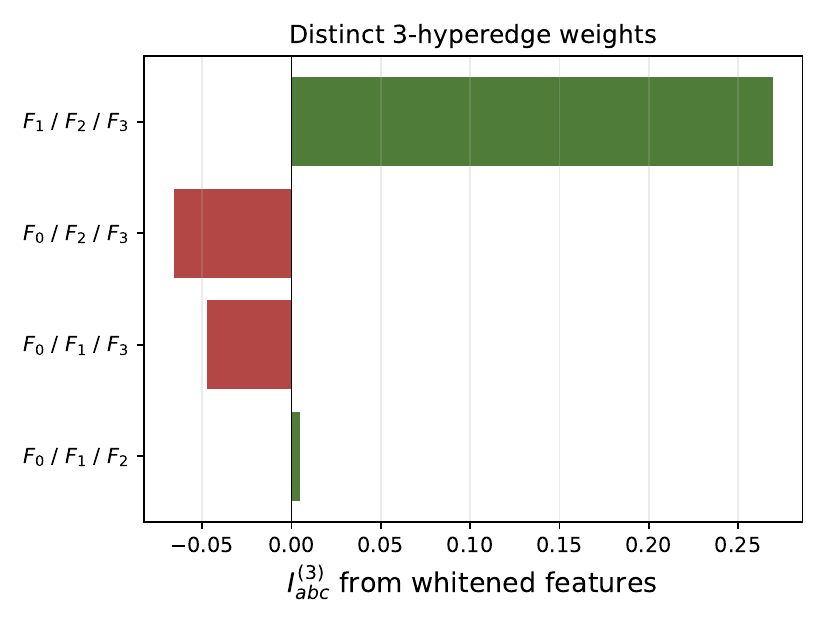}
\caption{Distinct 3-hyperedge weights in the whitened basis $F_0,\ldots,F_3$ defined in Table~\ref{tab:whitening_loadings}. Since whitening removes the non-trivial pairwise Fisher graph, the plotted coefficients are irreducible cubic Fisher weights. The dominant positive hyperedge $F_1/F_2/F_3$ is aligned, through the whitening loadings, with the wide-angle $\EEC$ and the second- and third-order $\mathrm{ECF}$ directions.}
\label{fig:hyperedge_weights}
\end{figure}

This compact study already closes the triangle. The cubic Fisher tensor is simultaneously a cubic correction to the local divergence, a connected three-observable cumulant of the correlator basis, and a 3-hyperedge weight used by the learning operator in Eq.~\eqref{eq:hgnn}. The same measured tensor therefore governs local inference, observable interpretation, and hypergraph propagation.

\subsection{Substructure Probes and Fisher-Hypergraph Compression}
\label{sec:substructure-application}

The second application asks whether the triality is sensitive to the prong substructure of hadronically decaying massive particles. Resolving the number of hard prongs inside a fat jet underlies essentially every modern boosted-object tagger, from dedicated substructure observables~\cite{Thaler:2010tr,Butterworth:2008iy} and shower deconstruction~\cite{Soper:2011cr} to large-parameter deep-learning architectures~\cite{Kasieczka:2019dbj}, and it is therefore a natural setting in which to ask whether the higher-order geometric structure captured by the triality carries discriminating information beyond what the Fisher graph alone encodes. The strategy is to construct two benchmark samples with cleanly separated hard substructure, estimate the Fisher information matrix and the leading non-trivial higher-order tensor, and convert them respectively into a graph and a hypergraph of observable correlations. We then select a basis of observables designed to probe two- and three-body structure together with the intra-jet radiation pattern~\cite{Thaler:2010tr,Komiske:2017aww,Larkoski:2017jix}, and retain the smallest sub-basis that preserves the structural divergence. The simplest physically motivated instance of this construction, a two-prong $W\to q\bar{q}$ decay versus a three-prong $t\to bq\bar{q}$ decay~\cite{Kaplan:2008ie,Kogler:2018hem}, is used to illustrate the procedure in the following sections.
\subsubsection{Benchmark and Compression Rule}

The benchmark dataset comprises $100{,}000$ jets each for the $W$ boson and hadronic top quark signal classes, clustered using the anti-$k_t$ algorithm~\cite{Cacciari:2008gp} with $R = 0.8$ (AK8). Jets are generated using the simulation setup described in Section~\ref{subsec:sim_setup}, with a transverse momentum requirement of $p_\mathrm{T} > 300\,\gev$ and mass-window selections centred on the respective $W$ boson and top quark masses. This sample serves as a concrete, physically motivated reference point at which the local Fisher tensors are estimated; the subsequent analysis should therefore be understood as a local compression benchmark evaluated on a realistic jet topology rather than a generalised information-geometric study. The observable basis is:
\begin{align}
    \phi_1 &= \EEC_{1}
    =
    \sum_{i<j} z_i z_j\, \mathbf{1}_{0.02 < \Delta R_{ij}\le 0.1},
    \\
    \phi_2 &= \EEC_{2}
    =
    \sum_{i<j} z_i z_j\, \mathbf{1}_{0.1 < \Delta R_{ij}\le 0.25},
    \\
    \phi_3 &= \EEC_{3}
    =
    \sum_{i<j} z_i z_j\, \mathbf{1}_{0.25 < \Delta R_{ij}\le 0.4},
    \\
    \phi_4 &= \EEC_{4}
    =
    \sum_{i<j} z_i z_j\, \mathbf{1}_{0.40 < \Delta R_{ij}\le 0.85},
    \\
    \phi_{5\dots13} &= e_2^{(\beta)},~\beta \in [0.25,0.5,0.75,1.0,1.25,1.5,2.0,2.5,3.0]
    \\
    \phi_{14\dots 19} &= e_3^{(\beta)},~\beta \in [0.5,1.0,1.25,1.5,2.0,2.5]
\end{align}
Rather than claim a realistic shower-level substructure tagger, we isolate the local compression problem itself by encoding the $W\to q\bar{q}$ versus $t\to bq\bar{q}$ distinction through a normalised benchmark score direction defined using the class-mean differences between the two-prong and three-prong samples
\begin{equation}
    \mathbf{u}
    =
    \mu_\mathrm{top}-\mu_W,
    \label{eq:prong-direction}
\end{equation}
where $\mu$ represents the mean of the 19-observable basis vector in the order written above. For a retained feature subset $S$, we define $u_S$ by setting the discarded components of $u$ to zero.

The pairwise graph baseline ranks observables by the quadratic node score
\begin{equation}
    s_a^{(2)}
    =
    \left|u_a \, \order{2}_{ab} u_b\right|,
    \label{eq:graph-score}
\end{equation}
while the multi-order Fisher hypergraph uses the additional cubic score
\begin{equation}
    s_a^{(3)}
    =
    \left|u_a \, \order{3}_{abc} u_b u_c\right|.
    \label{eq:hyper-score}
\end{equation}
We normalise both score vectors to unit sum,
\begin{equation}
    \widehat s_a^{(r)}
    =
    \frac{s_a^{(r)}}{\sum_b s_b^{(r)}},
    \qquad
    r=2,3,
\end{equation}
and define the direction-aligned multi-order node score
\begin{equation}
    s_a^{\mathrm{multi}}
    =
    \widehat s_a^{(2)}
    +
    \lambda_3 \widehat s_a^{(3)},
    \qquad
    \lambda_3=3.
    \label{eq:multi-score}
\end{equation}
In the present nineteen-observable benchmark, the two largest quadratic scores identify the pairwise-graph core
\begin{equation}
    S_0^{(2)}=\{e_2^{(1.25)},e_2^{(1)}\},
\end{equation}
while the multi-node score identifies the hypergraph core
\begin{equation}
    S_0^{(\mathrm{multi})}=\{e_3^{(1)},e_3^{(1.25)}\}.
\end{equation}
We use the node scores above for interpretation and for the left panel of Fig.~\ref{fig:selector_scores}, but the compression itself is performed by greedy forward selection. This makes the substructure sensitivity study structurally consistent with the basis-design benchmark in Section~\ref{sec:basis-application}.

For each selector we denote its starting core by $S_0$, with $S_0\equiv S_0^{(2)}$ for the pairwise-graph baseline and $S_0\equiv S_0^{(\mathrm{multi})}$ for the Fisher hypergraph. For a retained subset $S\supseteq S_0$ and a candidate observable $a\notin S$, define the mean pairwise redundancy
\begin{equation}
    \rho(a,S)
    =
    \frac{1}{|S|}
    \sum_{b\in S} \left|\order{2}_{ab}\right|,
\end{equation}
the redundancy-suppressed quadratic marginal gain
\begin{equation}
    \Delta_2(a\mid S)
    =
    \frac{
    \left|
    u_a^2 \order{2}_{aa}
    \right|
    +
    \left|
    2u_a \sum_{b\in S}\order{2}_{ab}u_b
    \right|
    }{
    1+\rho(a,S)
    },
    \label{eq:eft-delta2}
\end{equation}
and the aligned cubic completion gain
\begin{equation}
    \Delta_3(a\mid S)
    =
    \sum_{\{b,c\}\subseteq S}
    \left|
    u_a u_b u_c \order{3}_{abc}
    \right|.
    \label{eq:eft-delta3}
\end{equation}
To compare the quadratic and cubic gains on the same scale, we introduce the fixed rescaling factor
\begin{equation}
    r_3
    =
    \frac{
    \max_{a\notin S_0}\Delta_2(a\mid S_0)
    }{
    \max_{a\notin S_0}\Delta_3(a\mid S_0)
    }.
    \label{eq:eft-r3}
\end{equation}
The pairwise graph baseline then grows the basis by
\begin{equation}
    a^\star_{\mathrm{graph}}(S)
    =
    \operatorname*{arg\,max}_{a\notin S}
    \Delta_2(a\mid S),
    \label{eq:eft-graph-greedy}
\end{equation}
while the Fisher hypergraph uses
\begin{equation}
    a^\star_{\mathrm{hyper}}(S)
    =
    \operatorname*{arg\,max}_{a\notin S}
    \Big[
    \Delta_2(a\mid S)
    +
    \lambda_3 r_3 \Delta_3(a\mid S)
    \Big],
    \qquad
    \lambda_3=3.
    \label{eq:eft-hyper-greedy}
\end{equation}
With each selector starting from a core of two features, defined by the highest respective node scores,  each subsequent iteration is a non-trivial selection, and cubic triplet completion becomes relevant.

\subsubsection{Classification Performance Under Compression}

To compare the compressed observable bases selected by the two strategies, we train a Boosted Decision Tree (BDT) classifier of maximum depth 3 and 200 estimators on each retained subset, using the classification of jet events as arising from either a hadronic top quark decay ($t \to bq\bar{q}$, three-pronged) or a $W$ boson decay ($W \to q\bar{q}$, two-pronged) as the target task. Classification performance is quantified by the Area Under the ROC Curve (AUC) score, evaluated using a stratified $k$-fold cross-validation strategy with $k = 6$ folds, ensuring robust performance estimates against statistical fluctuations in the training sample. This protocol enables a direct, interpretable comparison of the substructure sensitivity retained by each compressed basis at each compression level.

The results are shown in Figure~\ref{fig:auc_wtop_study}. The hypergraph selection yields higher AUC scores throughout the compressed and intermediate regimes, demonstrating that the multi-node scoring strategy retains superior sensitivity to jet substructure when the observable basis is restricted to a small number of features. The performance gap is most pronounced at intermediate compression levels, where the two selection strategies diverge most strongly in their chosen observables, and closes once almost the full basis is retained. This improvement is not an artifact of the classifier or the training procedure, but reflects the intrinsic informativeness of the selected observables with respect to the prong multiplicity of the jet.

The origin of this gain is identified in the left plot of Figure~\ref{fig:selector_scores}, which contrasts the per-observable scores assigned by the pairwise (quadratic) and hypergraph selectors. The hypergraph scoring correctly elevates higher-order Energy Correlation Functions (ECFs) to the top of the ranked list, observables that the pairwise selector systematically undervalues. These higher-order ECFs encode irreducible three-point correlation information among jet constituents, which is precisely the structure that distinguishes a three-pronged hadronic top decay from the two-pronged $W$ topology. The pairwise selector, by construction sensitive only to marginal two-body correlations, cannot capture the collective discriminative value of such higher-order observables, and consequently creates a compressed basis that is structurally misaligned with the substructure signal. The hypergraph selector, by rewarding observables for their multi-point correlation content, recovers this information and translates it into measurable gains in classification performance.

\begin{figure}[t]
    \centering
    \includegraphics[width=0.98\textwidth]{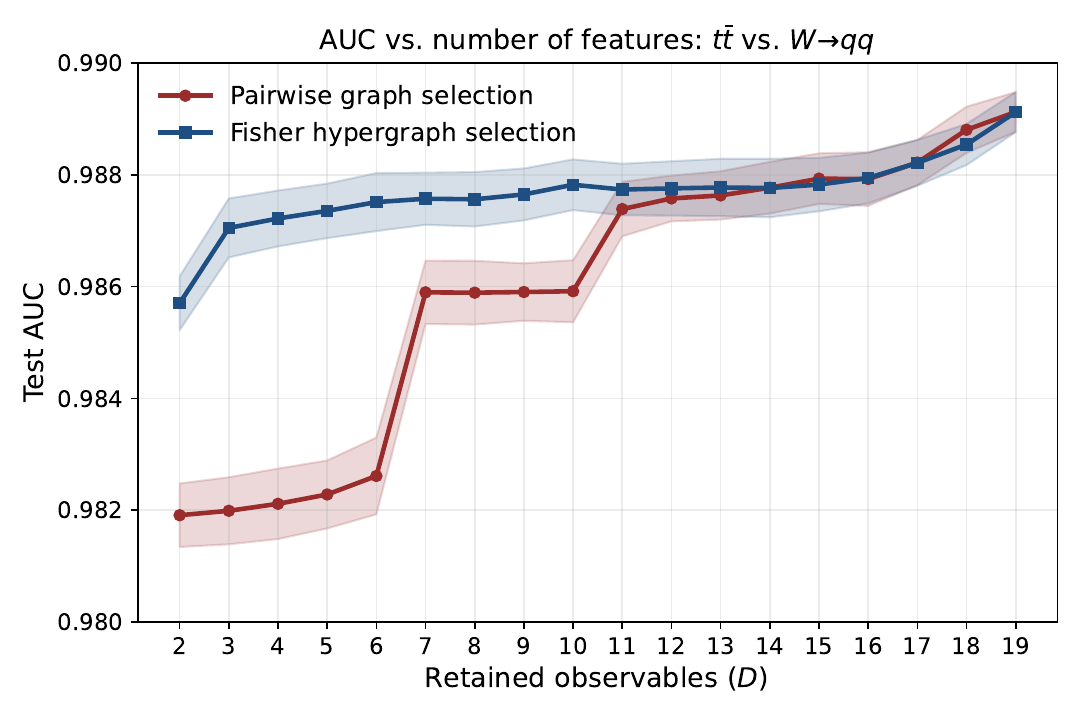}
    \caption{Classification performance for the $t\bar t$ versus $W\to q\bar q$ benchmark under greedy forward selection. The curves show mean test AUC as a function of the number of retained observables, with shaded bands indicating the spread across the cross-validation folds. The Fisher hypergraph selector performs better in the compressed regime, while the two selections coincide once almost the full basis is retained.}
    \label{fig:auc_wtop_study}
\end{figure}

The three legs of the triality are all used explicitly in this benchmark. On the Fisher side, $\order{2}$ and $\order{3}$ define the quadratic and cubic gains, and therefore the expected sensitivity to jet substructure. On the correlator side, those same coefficients are estimated from the EEC/ECF basis as connected fluctuations, so the selected structures can be read as concrete radiation patterns rather than as abstract scores. On the hypergraph side, $\order{2}$ supplies the pairwise baseline while $\order{3}$ supplies the 3-hyperedges whose completion changes the greedy selection. The compression gain, therefore, comes from a prong-sensitive pattern that is irreducibly triplet-valued.

\subsubsection{Interpretable Prong-Sensitive Hyperedges}

The compression gain is not a black-box effect. The right panel of Figure~\ref{fig:selector_scores} shows the direction-aligned node scores together with the leading 3-hyperedges ranked by
\begin{equation}
    \left|u_a u_b u_c \order{3}_{abc}\right|.
\end{equation}
The dominant triplet is $(e_3^{(1)},\,e_3^{(1.5)},\,e_3^{(2)})$, comprising three-point Energy Correlation Functions evaluated at angular exponents $\beta = 1,\,1.5,\,2$, followed by closely related triplets that couple two members of this family to lower-order two-point moments. This pattern carries a clear physical interpretation: the classification task is not driven by any single pairwise energy-flow observable, but by a collective signature of genuine three-pronginess. The $e_3^{(\beta)}$ family is, by construction, sensitive to three-body angular correlations among jet constituents, and its joint variation across multiple $\beta$ values provides a multi-scale characterisation of the three-prong geometry that is irreducible to any combination of two-point observables. Critically, this triplet structure is invisible to selection strategies based solely on the quadratic ($N=2$) Fisher information matrix, whose score functions reward pairwise discriminative power and cannot resolve the collective gain arising from the simultaneous inclusion of correlated three-point observables. It is only when the selection score is defined using higher-order information tensors at order $N > 2$, as in the hypergraph construction, that the greedy selector correctly identifies this triplet as the dominant informative structure, and consequently assembles a compressed basis that is geometrically aligned with the three-prong structure driving the $t \to bq\bar{q}$ versus $W \to q\bar{q}$ classification.

\begin{figure}[t]
    \centering
    \includegraphics[width=0.98\textwidth]{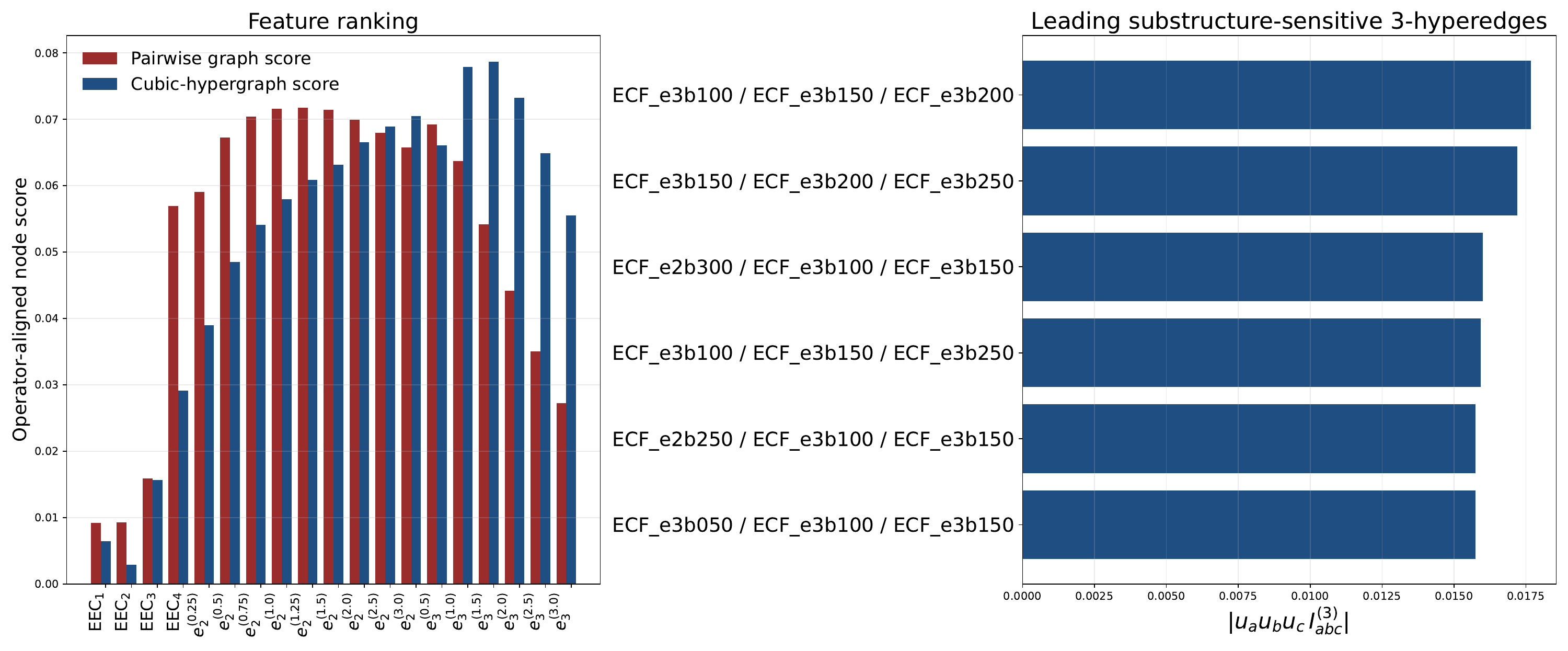}
    \caption{Left: direction-aligned feature scores for the pairwise graph and the multi-order Fisher hypergraph in the benchmark prong-sensitivity study. Right: the leading substructure-sensitive 3-hyperedges ranked by $|u_a u_b u_c\,\order{3}_{abc}|$. The dominant hyperedge ties together the three three-point energy correlators.}
    \label{fig:selector_scores}
\end{figure}

This shows the practical use of the triality. The same tensor simultaneously quantifies local distinguishability, identifies the physically important correlator triplets, and prescribes the higher-order combinatorial structure on which a learning or compression algorithm should operate. While a more realistic analysis with more complex algorithms could be designed, this simple benchmark demonstrates how the full chain would work in practice. The minimal class-mixture study in Section~\ref{sec:toy-application} provides a complementary low-dimensional illustration.

\subsection{Observable Discovery and Higher-Order Basis Design}
\label{sec:basis-application}

A third application is to the design of the observable basis itself. In a realistic jet analysis one often begins with tens or hundreds of binned EECs together with a selected ECF/EFP basis, and the main practical question is not how to fit one fixed set of observables but which observables should survive into the final analysis. This is closely related to information-theoretic approaches to observable design \cite{Larkoski:2014pca}. The triality provides a principled answer because it attaches a local statistical meaning to every connected correlator pattern. Given a task direction $\Delta^a$ in feature space, one may rank hyperedges by the aligned score
\begin{equation}
    S_e^{(n)}
    =
    \left|
    \Delta^{a_1}\cdots \Delta^{a_n}
    \order{n}_{a_1\cdots a_n}
    \right|,
    \qquad
    e=\{a_1,\ldots,a_n\}.
    \label{eq:aligned-score}
\end{equation}
Because the same coefficient is also a connected correlator cumulant, a large value of $S_e^{(n)}$ identifies not merely a large covariance but an irreducible multipoint radiation pattern that matters for the task at hand. This is precisely the information that is lost if one ranks observables by pairwise correlations alone.

To make this concrete we now consider a higher-dimensional benchmark. The reference sample is the boosted top hadronic decay jet (three-pronged) used in the substructure study of Section~\ref{sec:substructure-application}, but the basis-design question is now posed against a \emph{family} of nearby deformations of the reference sample rather than a clear and easily resolvable difference in the jet substructure itself. We use three benchmark deformations with clear collider interpretations: an opening-angle shift of the hard prongs, an inter-prong bridge component and a soft fourth-prong contamination. The candidate observable library contains eighteen binned EECs spanning $0.02<\Delta R\le 0.85$, nine two-point ECFs, and six three-point ECFs,
\begin{equation}
    \big\{\EEC_1,\ldots,\EEC_{18},\ e_2^{(0.25)},\ e_2^{(0.5)},\ldots,e_2^{(3)},\ e_3^{(0.5)},\ldots,e_3^{(3)}\big\}.
\end{equation}
For each deformation $m=1,\ldots,N_\Delta$ with $N_\Delta=3$, let $\Delta_{(m)}^a$ denote the standardised mean-shift direction relative to the reference sample. We construct three benchmark deformations of the reference $t\bar{t}$ distribution
through phase-space selections on per-jet substructure features.
For each jet, we compute three hard prong axes $\{a_1, a_2, a_3\}$ via
exclusive $k_T$ reclustering with $N_\mathrm{exclusive} = 3$, and define
$z_i = p_T^{a_i} / \sum_k p_T^{a_k}$, the pairwise opening angles
$\Delta R_{ij}$, and the N-subjettiness ratio $\tau_4/\tau_3$.
Each deformation $k$ is realised by selecting the subset $S_k$ of events
lying in the upper (or lower) tail fraction $q \in (0,1)$ of a deformation
score $s_k$, as summarised in Table~\ref{tab:deformations}.
The corresponding mean-shift direction in the 33-dimensional observable space
is then
\begin{equation}
  \Delta^{(k)} = \langle \theta \rangle_{S_k} - \langle \theta \rangle_\mathrm{ref},
\end{equation}
where $\langle O \rangle_\mathrm{ref}$ is computed over the remaining background
sample. We use $q=0.1$ as a reasonable choice. 

\begin{table}[htbp]
\centering
\begin{tabular}{lll}
\hline
Deformation & Score $s_k$ & Tail \\
\hline
Opening-angle shift      & $\Delta R_{12}$                                                              & top $q$ \\
Inter-prong bridge       & $\sum_j p_T^j\, \mathbf{1}(R_\mathrm{col} < \Delta R_{j,\mathrm{axis}} < R_\mathrm{wide})$ & top $q$ \\
Soft 4th-prong contamination & $\tau_4/\tau_3$                                                        & bottom $q$ \\
\hline
\end{tabular}
\caption{Deformation scores and tail selections used to construct the
benchmark mean-shift directions. Default radii: $R_\mathrm{wide} = 0.4$,
$R_\mathrm{col} = 0.15$.}
\label{tab:deformations}
\end{table}

A basis-design algorithm should suppress redundancy: once a retained subset already captures a given pairwise fluctuation pattern, adding another observable from the same pairwise ladder should produce only a small marginal gain. We therefore replace static top-$k$ ranking by a greedy forward selection rule. For a retained subset $S$, let $\Delta_{(m),S}$ denote the direction obtained from $\Delta_{(m)}$ by setting the discarded components to zero, and define the retained quadratic Fisher response
\begin{equation}
    Q(S)
    =
    \frac{1}{N_\Delta}
    \sum_{m=1}^{N_\Delta}
    \frac{
    \Delta_{(m),S}^a \order{2}_{ab} \Delta_{(m),S}^b
    }{
    \Delta_{(m)}^a \order{2}_{ab} \Delta_{(m)}^b
    }.
    \label{eq:quadratic-basis-objective}
\end{equation}
Because Eq.~\eqref{eq:quadratic-basis-objective} is evaluated on the current subset rather than on isolated node scores, it penalises redundant pairwise picks automatically. The pairwise graph baseline therefore grows the basis by the marginal-gain rule
\begin{equation}
    a^\star_{\mathrm{graph}}(S)
    =
    \operatorname*{arg\,max}_{a\notin S}
    \Big[
    Q(S\cup\{a\})-Q(S)
    \Big].
    \label{eq:graph-greedy}
\end{equation}
To track genuinely irreducible three-way structure, we complement $Q(S)$ by the aligned triplet coverage
\begin{equation}
    C_3^{\mathrm{align}}(S)
    =
    \frac{
    \sum_{m=1}^{N_\Delta}
    \sum_{\{i,j,k\}\subseteq S}
    \left|
    \Delta_{(m)}^i \Delta_{(m)}^j \Delta_{(m)}^k
    \order{3}_{ijk}
    \right|
    }{
    \sum_{m=1}^{N_\Delta}
    \sum_{i<j<k}
    \left|
    \Delta_{(m)}^i \Delta_{(m)}^j \Delta_{(m)}^k
    \order{3}_{ijk}
    \right|
    }.
    \label{eq:cubic-coverage}
\end{equation}
The Fisher-hypergraph selector uses the same redundancy-suppressing greedy search, but augments the quadratic gain by the aligned cubic coverage,
\begin{equation}
    a^\star_{\mathrm{hyper}}(S)
    =
    \operatorname*{arg\,max}_{a\notin S}
    \Big[
    Q(S\cup\{a\})-Q(S)
    +
    \lambda_3
    \big(
    C_3^{\mathrm{align}}(S\cup\{a\})-C_3^{\mathrm{align}}(S)
    \big)
    \Big],
    \qquad
    \lambda_3=\frac32.
    \label{eq:hyper-greedy}
\end{equation}
The two selectors therefore differ only in whether the cubic triality information is allowed to influence the marginal gain.

The same three ingredients are therefore active throughout this basis-design problem. The Fisher tensors provide the objectives: $Q(S)$ measures retained quadratic response, while $C_3^{\mathrm{align}}(S)$ measures retained rank-three response across the deformation family. Because these same tensors are also connected correlator cumulants, a large cubic marginal gain identifies a genuine multi-observable radiation pattern rather than a redundant refinement of an already covered pairwise ladder. The hypergraph language then turns those same cubic coefficients into triplets that the selector tries to close, which is why the final basis favors wide-angle and three-point structure over extra short-distance pairwise bins.

Because the aim of this subsection is basis design across several nearby tasks, we evaluate each compressed basis using the deformation directions, assessing the final information retention with the corresponding Fisher-aligned natural-parameter direction. For each deformation we measure the retained local third-order divergence fraction

\begin{equation}
    R_{(3)}^{(m)}(S) \;=\; \frac{\tfrac{1}{2}\,\Theta_{(m),S}^{\top} \mathcal{I}^{(2)}_{S}\,\Theta_{(m),S} \;+\; \tfrac{1}{6}\,\mathcal{I}^{(3)}_{S,ijk}\,\Theta_{(m),S}^{i}\Theta_{(m),S}^{j}\Theta_{(m),S}^{k}}{\tfrac{1}{2}\,\Theta^{\top}_{(m)}\mathcal{I}^{(2)}\,\Theta_{(m)} \;+\; \tfrac{1}{6}\,\mathcal{I}^{(3)}_{ijk}\,\Theta^{i}_{(m)}\Theta^{j}_{(m)}\Theta^{k}_{(m)}},
    \label{eq:r3}
\end{equation}
where:
\begin{equation*}
    \Theta_{(m),S} = {\mathcal{I}^{(2)}_{S}}^{-1}\Delta_S,\ \ \Theta_{(m)} = {\mathcal{I}^{(2)}}^{-1}\Delta_{(m)}.
\end{equation*}
 
We summarise this family of tasks by its mean and median,
\begin{equation}
    \overline{R}_{\mathrm{(3)}}(S)
    =
    \frac{1}{N_\Delta}
    \sum_{m=1}^{N_\Delta}
    R_{\mathrm{(3)}}^{(m)}(S),
    \qquad
    R_{\mathrm{(3)}}^{\mathrm{med}}(S)
    =
    \mathrm{median}_{m}\, R_{\mathrm{(3)}}^{(m)}(S).
    \label{eq:r3_metrics}
\end{equation}

Figure~\ref{fig:basis-design} shows the result. Once redundancy is handled at the selection stage, the pairwise graph and Fisher hypergraph no longer collapse onto nearly identical top-$k$ lists. The two selectors coincide only for the first few obvious picks and then separate over the full range $k=6$--$18$. The hypergraph basis improves the exact local diagnostics at second and third order, and the aligned triplet coverage simultaneously throughout that regime. At twelve retained observables, the pairwise graph keeps $\overline{R}_{\mathrm{(3)}}=0.871$, $R_{\mathrm{(3)}}^{\mathrm{med}}=0.899$, and $C_3^{\mathrm{align}}=0.449$, whereas the Fisher hypergraph keeps $0.937$, $0.938$, and $0.490$, respectively. The correct interpretation is therefore that the observed gain is not only a change in ordering: but also that the hypergraph selector retains more of the local Fisher/KL information while also preserving more deformation-aligned cubic structure. This supports using the hypergraph selector as the preferred means of constructing compressed bases for an analysis with a large number of possible observables to choose from, correctly maximising the sensitivity to deformations of the reference sample space.

\begin{figure}[t]
    \centering
    \includegraphics[width=0.98\textwidth]{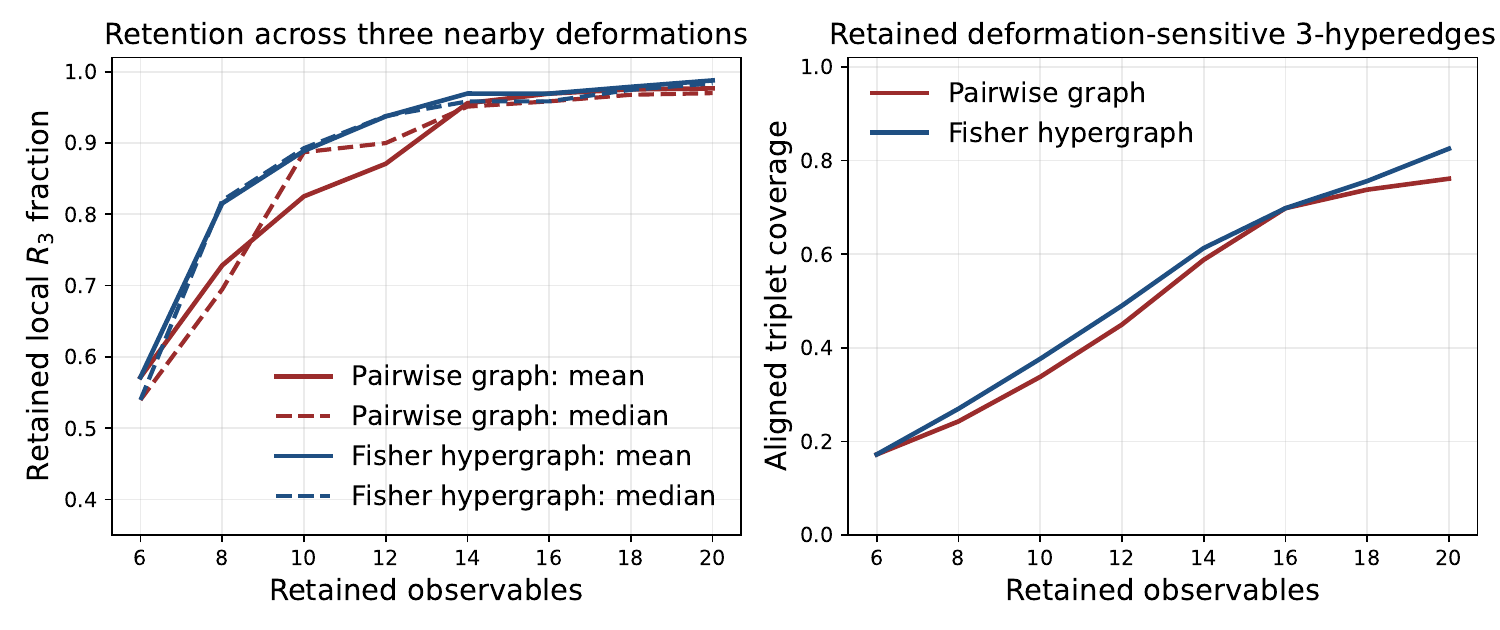}
    \caption{Greedy basis design across three nearby radiation-pattern deformations in a $33$-observable EEC/ECF library. Left: mean and median retained local $R_{(3)}$ fractions, evaluated with Eq.~\eqref{eq:r3}. Right: aligned triplet coverage from Eq.~\eqref{eq:cubic-coverage}. Once redundancy is suppressed by greedy marginal-gain selection, the Fisher hypergraph keeps more exact local discrimination power and visibly more deformation-sensitive rank-three structure over a broad range of feature budgets.}
    \label{fig:basis-design}
\end{figure}

\subsection{Physics-Informed Hypergraph Learning}
\label{sec:learning-application}

The fourth application turns Eq.~\eqref{eq:hgnn} into a concrete low-capacity learning benchmark. Graph-based collider architectures such as ParticleNet and energy-weighted message passing have already shown the effectiveness of physics-aware message passing on jet data \cite{Qu:2019gqs,Konar:2021zdg}. For the present purpose a nearby-deformation task is more informative than the QCD-versus-top separation of Section~\ref{sec:toy-application}, because in the four-feature toy benchmark even simple linear classifiers saturate quickly. We therefore return to the $33$-observable EEC/ECF library and the three-prong reference family of Section~\ref{sec:basis-application}, and classify the reference sample against the soft fourth-prong deformation. This is precisely the regime where the graph and hypergraph priors differ most. In a representative reference training sample the leading pairwise Fisher graph is dominated by the adjacent $e_2$ ladder, whereas the leading 3-hyperedges are $(e_3^{(1.5)},e_3^{(2)},e_3^{(3)})$, $(e_3^{(1)},e_3^{(2)},e_3^{(3)})$, and $(e_3^{(1)},e_3^{(1.5)},e_3^{(3)})$.

For an event represented by the standardised node-feature vector $x\in \mathbb{R}^{33}$, we keep the order-specific propagators fixed and train only a one-layer readout,

\begin{equation}
    h(x)
    =
    \sigma\!\left(
    x U_0 + U_2(\mathcal{P}_2 x)  + U_3(\mathcal{P}_3 x)
    \right),
    \qquad
    s(x)
    =
    \beta^\top \bar h(x),
    \qquad
    \bar h(x)=\frac{1}{|V|}\sum_{a\in V} h_a(x),
    \label{eq:learning-benchmark}
\end{equation}

with the fixed order-$r$ propagation matrices
\begin{equation}
    \mathcal{P}_r
    =
    \left(D_v^{(r)}\right)^{-1/2}
    B^{(r)} |W^{(r)}|
    \left(D_e^{(r)}\right)^{-1}
    \left(B^{(r)}\right)^\top
    \left(D_v^{(r)}\right)^{-1/2}.
\end{equation}
$U_2$ and $U_3$ are simple multi-layer perceptrons (MLPs) with a single hidden layer and about $500$ trainable parameters each. The pairwise graph baseline sets $U_3=0$. The triality-informed hypergraph keeps both orders, with the order-two layer built from the top $150$ entries of $|\order{2}_{ij}|$ and the order-three layer from the top $200$ entries of $|\order{3}_{ijk}|$ estimated on the reference training sample. 

The benchmark uses balanced training samples of $3000$, $6000$, and $9000$ labelled events, together with fixed validation and test sets, and Figure~\ref{fig:learning-summary} reports means and standard deviations over five random splits.

The result is modest in absolute size, because the classifier is intentionally tiny, but it is systematic. At $3000$ labelled training events the Fisher hypergraph reaches a mean test AUC of $0.8163\pm 0.0068$, compared with $0.8102\pm 0.0061$ for the pairwise graph. At $6000$ labelled events the same ordering persists, with $0.8258\pm 0.0029$ for the hypergraph versus $0.8184\pm 0.0039$. The right panel of Figure~\ref{fig:learning-summary} shows the same effect dynamically: at fixed model size the triality-informed hypergraph converges to a higher validation AUC than the graph baseline. Therefore, once the signal depends on a subtle deformation of the radiation pattern rather than on an easy top-vs-QCD split, the Fisher 3-hyperedges become an informative inductive bias rather than a decorative add-on.

This benchmark uses all three legs of the triality in a direct way. The Fisher tensors supply the pairwise and triplet weights that define the two propagation channels. The correlator hierarchy explains why the informative 3-hyperedges sit in the high-$\beta$ $e_3$ sector rather than in another redundant refinement of the pairwise $e_2$ ladder. The hypergraph layer then turns that information-geometric and correlator structure into a concrete learning prior. The gain in Figure~\ref{fig:learning-summary} is therefore not a generic statement that ``hypergraphs are better''; it is the practical benefit of using a hypergraph whose topology and weights have been fixed by the local Fisher--correlator data.

\begin{figure}[t]
    \centering
    \includegraphics[width=\textwidth]{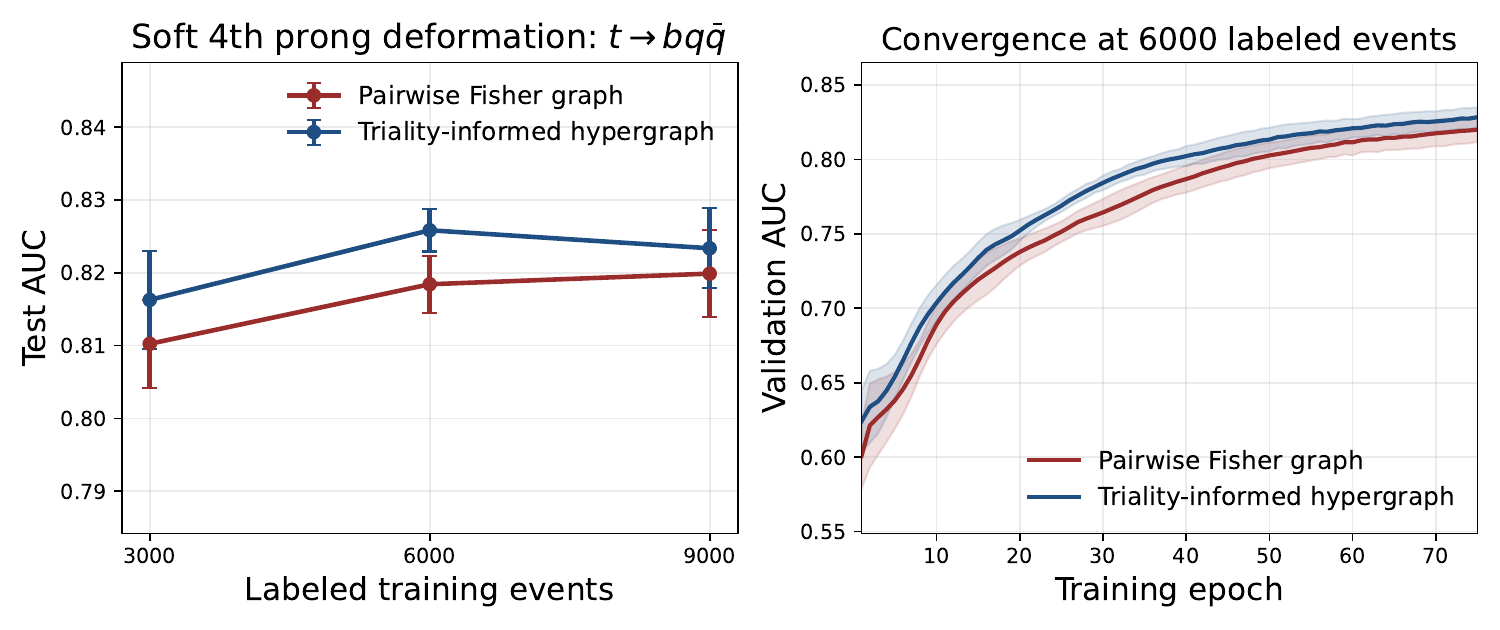}
    \caption{Physics-informed hypergraph learning for the soft fourth-prong deformation in the $33$-observable EEC/ECF library. Left: mean test AUC versus the number of labelled training events for a pairwise Fisher graph and a triality-informed Fisher hypergraph; error bars show one standard deviation over five random splits. Right: mean validation AUC versus epoch at $6000$ labelled training events, with shaded bands showing one standard deviation over the same splits. The faster convergence and higher validation score of the triality-informed hypergraph are consistent with retaining irreducible triplet structure that the graph baseline smears out.}
    \label{fig:learning-summary}
\end{figure}

\section{Summary and Conclusions}
\label{sec:outlook}

We have shown that, once a collider event family is represented in the natural exponential-family coordinates associated with a chosen finite basis of observables, higher Fisher tensors, connected correlator cumulants, and weighted hyperedges are the same local tensor written in three different languages. In this representation the log-partition function generates the full hierarchy, so that the coefficients governing the local Kullback--Leibler expansion also determine the connected multi-observable radiation structure and the weights of the corresponding hyperedges. This provides an exact Fisher--correlator--hypergraph triality for a finite basis of EECs, ECFs, and EFPs.

The scope of this statement is important. The identification is coordinate-specific: it is exact in the canonical chart generated by the chosen sufficient statistics. In a generic parametrisation the local information geometry remains well defined, but from cubic order onwards the KL coefficients contain additional Hessian contributions, as shown in Appendix~\ref{app:generic}, and no longer coincide directly with connected observable cumulants. For applications this does not require the full collider problem to be globally exponential-family. It requires a local embedding around a reference sample in which the chosen observable basis captures the deformations relevant to the task.

The four applications show that the cubic level of this hierarchy is already useful in practice. In the QCD-versus-top local-KL study of Section~\ref{sec:toy-application}, the cubic tensor improves the local KL approximation and isolates a dominant irreducible triplet coupling wide-angle EEC support to the two- and three-point ECFs. In the prong-substructure compression benchmark of Section~\ref{sec:substructure-application}, the hypergraph selector improves BDT classification of $W\to q\bar{q}$ versus $t\to bq\bar{q}$ jets in the compressed regime by elevating higher-order ECFs associated with the leading three-point correlator triplets.

In the three-deformation basis-design study of Section~\ref{sec:basis-application}, the same cubic information improves compressed observable selection. At twelve observables the Fisher hypergraph reaches $\overline{R}_{(3)}=0.937$, $R_{(3)}^{\mathrm{med}}=0.938$, and $C_3^{\mathrm{align}}=0.490$, compared with $0.871$, $0.899$, and $0.449$ for the pairwise graph. In the physics-informed learning benchmark of Section~\ref{sec:learning-application}, a one-layer classifier initialised from the pairwise Fisher graph and Fisher 3-hypergraph reaches a mean test AUC of $0.8258\pm0.0029$ at $6000$ labelled events for the soft fourth-prong deformation, slightly above the pairwise graph baseline of $0.8184\pm 0.0039$ at the $\sim\!1.5\,\sigma$ level over five splits.

Taken together, these results show that the benefit of the triality is not merely that one can attach a hypergraph interpretation to a higher Fisher tensor. The main point is that the same measured or simulated correlator data can carry a common quantitative interpretation across local inference, physics interpretation, basis compression, and learning architecture design. Once a large cubic coefficient is identified, one learns that a specific three-observable radiation pattern is genuinely connected, that it contributes to local distinguishability beyond the Fisher-matrix approximation, and that the corresponding triplet should be retained or propagated jointly in a hypergraph representation. In this way the triality turns higher Fisher tensors from abstract information-geometric objects into quantities that guide which observables to keep, which multi-observable structures to interpret as physically irreducible, and how to initialise higher-order architectures in a way tied directly to the underlying radiation pattern.

Several extensions are natural. For collider analyses, the triality suggests new strategies for analysis-basis construction, feature compression, and the interpretation of wide-angle and multi-prong radiation patterns in both precision and new-physics searches \cite{Nachman:2020ccu,Blance:2019ibf}. For machine learning, it provides a physics-motivated recipe for building higher-order architectures whose topology and weights are initialised from measured or simulated local structure rather than from generic architectural heuristics \cite{Amari:1998,Qu:2019gqs,Konar:2021zdg}. For information geometry, collider observables offer a concrete setting in which higher Fisher tensors have an immediate operational interpretation as local divergence coefficients, connected cumulants, and combinatorial weights. On the algorithmic side, the main open questions are how best to truncate, sparsify, regularise, and propagate the higher-order hierarchy in larger models and larger observable libraries.

\appendix

\section{Generic-Coordinate Form of the Cubic KL Coefficient}
\label{app:generic}

In the main text the triality is formulated in natural exponential-family coordinates, where derivatives of the log-partition function generate both the local KL coefficients and the connected cumulants of the chosen observable basis. Here, we show what changes when one works in a completely generic parametrisation of the same family of probability distributions.

The key point is that the Fisher matrix at quadratic order is special: it can still be written in a clean way in generic coordinates. Starting at cubic order, however, the coefficients of the local KL expansion are no longer determined only by moments of the score. They also contain terms involving second derivatives of the log-likelihood, which measure how the chosen coordinate chart itself bends through the statistical model. The exact equality used in the main text is therefore not a generic coordinate identity, but a property of the natural exponential-family representation.

Let $\ell_\theta(x)=\log p_\theta(x)$ and $s_a=\partial_a \ell_\theta$. Expanding the KL divergence around $\delta=0$ gives
\begin{equation}
    \KL\!\left(p_\theta \Vert p_{\theta+\delta}\right)
    =
    -\sum_{n=1}^{\infty}\frac{1}{n!}\,
    \avg{\partial_{a_1}\cdots\partial_{a_n}\ell_\theta}_\theta\,
    \delta^{a_1}\cdots\delta^{a_n}.
    \label{eq:generic-kl-expand}
\end{equation}
Equation~\eqref{eq:generic-kl-expand} is completely general. It says that the local KL coefficients are expectation values of higher derivatives of the log-likelihood. The first term vanishes because the KL divergence is minimised at $\delta=0$, and the quadratic term reduces to the usual Fisher information matrix after the standard normalisation identities are used. Those identities are
\begin{align}
    0 &= \avg{s_a}_\theta, \\
    0 &= \avg{\partial_{ab}\ell_\theta + s_a s_b}_\theta.
\end{align}
The first relation is simply the statement that the score has zero mean. The second says that the expected Hessian of the log-likelihood is the negative Fisher matrix. Up to this order, the expansion behaves in the familiar way. The first place where the coordinate dependence becomes non-trivial is the cubic coefficient. Differentiating the second identity once more yields
\begin{equation}
    0
    =
    \avg{
    \partial_{abc}\ell_\theta
    + s_a \partial_{bc}\ell_\theta
    + s_b \partial_{ac}\ell_\theta
    + s_c \partial_{ab}\ell_\theta
    + s_a s_b s_c
    }_\theta .
    \label{eq:cubic-identity}
\end{equation}
This identity shows that the third derivative of the log-likelihood is not determined by the score-cubic moment alone, but mixes with products of scores and Hessians. Therefore the cubic KL coefficient in generic coordinates is
\begin{equation}
    \frac{1}{6}
    \avg{
    s_a s_b s_c
    + s_a \partial_{bc}\ell_\theta
    + s_b \partial_{ac}\ell_\theta
    + s_c \partial_{ab}\ell_\theta
    }_\theta
    \delta^a\delta^b\delta^c.
    \label{eq:generic-cubic}
\end{equation}
Equation~\eqref{eq:generic-cubic} reveals that the cubic local KL coefficient generally contains two kinds of information at once. One part is the genuine third-order fluctuation of the score, which is the quantity that becomes the connected observable cumulant in the exponential-family setting. The other part comes from the curvature of the chosen parametrisation, encoded in the Hessian terms. Only in natural exponential coordinates does that second contribution reduce in such a way that the cubic KL coefficient matches the observable cumulant directly. This is why the triality theorem in the main text is stated in the exponential-family embedding rather than for a completely arbitrary parametrisation.

Equivalently, one may say that natural exponential coordinates are the chart in which the local statistical geometry is generated directly by the sufficient statistics. In that chart the cubic coefficient has a clean physical interpretation as a connected three-observable fluctuation and, at the same time, as a hyperedge weight. In a generic chart the same local distinguishability is still present, but the decomposition into physically interpretable cumulants and purely coordinate-dependent terms is mixed. This is the precise sense in which the triality is exact only in the exponential-family representation.

\acknowledgments
The authors acknowledge the support of Schmidt Sciences, the Alexander von Humboldt Foundation, and IPPP Durham, as well as the usage of computing resources at the Institute of Experimental Particle Physics (ETP), KIT.

\bibliographystyle{unsrtnat}
\bibliography{apssamp}

@book{Amari:2000,
  author = {Amari, Shun-ichi and Nagaoka, Hiroshi},
  title = {Methods of Information Geometry},
  publisher = {American Mathematical Society and Oxford University Press},
  year = {2000},
  series = {Translations of Mathematical Monographs},
  volume = {191}
}

@book{Ay:2017,
  author = {Ay, Nihat and Jost, J{\"u}rgen and Le, H{\^o}ng V{\^a}n and Schwachh{\"o}fer, Lorenz},
  title = {Information Geometry},
  series = {Ergebnisse der Mathematik und ihrer Grenzgebiete. 3. Folge},
  volume = {64},
  publisher = {Springer},
  year = {2017},
  doi = {10.1007/978-3-319-56478-4}
}

@book{Bretto:2013,
  author = {Bretto, Alain},
  title = {Hypergraph Theory: An Introduction},
  series = {Mathematical Engineering},
  publisher = {Springer},
  year = {2013},
  doi = {10.1007/978-3-319-00080-0}
}

@article{Eguchi:1985,
  author = {Eguchi, Shinto},
  title = {A differential geometric approach to statistical inference on the basis of contrast functionals},
  journal = {Hiroshima Mathematical Journal},
  volume = {15},
  pages = {341--391},
  year = {1985},
  doi = {10.32917/hmj/1206130775}
}

@article{Eguchi:1992,
  author = {Eguchi, Shinto},
  title = {Geometry of minimum contrast},
  journal = {Hiroshima Mathematical Journal},
  volume = {22},
  number = {3},
  pages = {631--647},
  year = {1992},
  doi = {10.32917/hmj/1206128508}
}

@article{Basham:1978zq,
    author = "Basham, C. L. and Brown, L. S. and Ellis, S. D. and Love, S. T.",
    title = "{Energy Correlations in electron-Positron Annihilation in Quantum Chromodynamics: Asymptotically Free Perturbation Theory}",
    reportNumber = "RLO-1388-761",
    doi = "10.1103/PhysRevD.19.2018",
    journal = "Phys. Rev. D",
    volume = "19",
    pages = "2018",
    year = "1979"
}

@article{Larkoski:2013eya,
    author = "Larkoski, Andrew J. and Salam, Gavin P. and Thaler, Jesse",
    title = "{Energy Correlation Functions for Jet Substructure}",
    eprint = "1305.0007",
    archivePrefix = "arXiv",
    primaryClass = "hep-ph",
    reportNumber = "MIT-CTP-4446, CERN-PH-TH-2013-066, LPN13-026",
    doi = "10.1007/JHEP06(2013)108",
    journal = "JHEP",
    volume = "06",
    pages = "108",
    year = "2013"
}

@article{Komiske:2017aww,
    author = "Komiske, Patrick T. and Metodiev, Eric M. and Thaler, Jesse",
    title = "{Energy flow polynomials: A complete linear basis for jet substructure}",
    eprint = "1712.07124",
    archivePrefix = "arXiv",
    primaryClass = "hep-ph",
    reportNumber = "MIT-CTP-4965",
    doi = "10.1007/JHEP04(2018)013",
    journal = "JHEP",
    volume = "04",
    pages = "013",
    year = "2018"
}

@article{Hofman:2008ar,
    author = "Hofman, Diego M. and Maldacena, Juan",
    title = "{Conformal collider physics: Energy and charge correlations}",
    eprint = "0803.1467",
    archivePrefix = "arXiv",
    primaryClass = "hep-th",
    doi = "10.1088/1126-6708/2008/05/012",
    journal = "JHEP",
    volume = "05",
    pages = "012",
    year = "2008"
}

@article{Moult:2018jzp,
    author = "Moult, Ian and Zhu, Hua Xing",
    title = "{Simplicity from Recoil: The Three-Loop Soft Function and Factorization for the Energy-Energy Correlation}",
    eprint = "1801.02627",
    archivePrefix = "arXiv",
    primaryClass = "hep-ph",
    doi = "10.1007/JHEP08(2018)160",
    journal = "JHEP",
    volume = "08",
    pages = "160",
    year = "2018"
}

@inproceedings{Zhou:2006,
  author = {Zhou, Dengyong and Huang, Jiayuan and Sch{\"o}lkopf, Bernhard},
  title = {Learning with Hypergraphs: Clustering, Classification, and Embedding},
  booktitle = {Advances in Neural Information Processing Systems 19},
  year = {2007}
}

@article{DBLP:journals/corr/abs-1809-09401,
  author       = {Yifan Feng and
                  Haoxuan You and
                  Zizhao Zhang and
                  Rongrong Ji and
                  Yue Gao},
  title        = {Hypergraph Neural Networks},
  journal      = {CoRR},
  volume       = {abs/1809.09401},
  year         = {2018},
  url          = {http://arxiv.org/abs/1809.09401},
  eprinttype   = {arXiv},
  eprint       = {1809.09401},
  bibsource    = {dblp computer science bibliography, https://dblp.org}
}

@article{Birch-Sykes:2024gij,
    author = "Birch-Sykes, Callum and Le, Brian and Peters, Yvonne and Simpson, Ethan and Zhang, Zihan",
    title = "{Reconstructing short-lived particles using hypergraph representation learning}",
    eprint = "2402.10149",
    archivePrefix = "arXiv",
    primaryClass = "hep-ph",
    doi = "10.1103/PhysRevD.111.032004",
    journal = "Phys. Rev. D",
    volume = "111",
    number = "3",
    pages = "032004",
    year = "2025"
}

@article{Konar:2023ptv,
    author = "Konar, Partha and Ngairangbam, Vishal S. and Spannowsky, Michael",
    title = "{Hypergraphs in LHC phenomenology {\textemdash} the next frontier of IRC-safe feature extraction}",
    eprint = "2309.17351",
    archivePrefix = "arXiv",
    primaryClass = "hep-ph",
    reportNumber = "IPPP/23/53",
    doi = "10.1007/JHEP01(2024)113",
    journal = "JHEP",
    volume = "01",
    pages = "113",
    year = "2024"
}

@article{Rao:1945,
  author = {Rao, C. Radhakrishna},
  title = {Information and the Accuracy Attainable in the Estimation of Statistical Parameters},
  journal = {Bull. Calcutta Math. Soc.},
  volume = {37},
  pages = {81--91},
  year = {1945}
}

@book{Chentsov:1982,
  author = {Cencov, N. N.},
  title = {Statistical Decision Rules and Optimal Inference},
  publisher = {American Mathematical Society},
  series = {Translations of Mathematical Monographs},
  volume = {53},
  address = {Providence, RI},
  year = {1982},
  note = {English translation of the 1972 Russian original}
}

@book{Amari:1985,
  author = {Amari, Shun-ichi},
  title = {Differential-Geometrical Methods in Statistics},
  series = {Lecture Notes in Statistics},
  volume = {28},
  publisher = {Springer},
  year = {1985},
  doi = {10.1007/978-1-4612-5056-2}
}

@book{BarndorffNielsen:1978,
  author = {Barndorff-Nielsen, Ole E.},
  title = {Information and Exponential Families in Statistical Theory},
  publisher = {Wiley},
  series = {Wiley Series in Probability and Mathematical Statistics},
  year = {1978}
}

@book{McCullagh:1987,
  author = {McCullagh, Peter},
  title = {Tensor Methods in Statistics},
  publisher = {Chapman and Hall},
  year = {1987}
}

@article{Dixon:2019uzg,
    author = "Dixon, Lance J. and Moult, Ian and Zhu, Hua Xing",
    title = "{Collinear limit of the energy-energy correlator}",
    eprint = "1905.01310",
    archivePrefix = "arXiv",
    primaryClass = "hep-ph",
    reportNumber = "SLAC-PUB-17427",
    doi = "10.1103/PhysRevD.100.014009",
    journal = "Phys. Rev. D",
    volume = "100",
    number = "1",
    pages = "014009",
    year = "2019"
}

@article{Chen:2019bpb,
    author = "Chen, Hao and Luo, Ming-Xing and Moult, Ian and Yang, Tong-Zhi and Zhang, Xiaoyuan and Zhu, Hua Xing",
    title = "{Three point energy correlators in the collinear limit: symmetries, dualities and analytic results}",
    eprint = "1912.11050",
    archivePrefix = "arXiv",
    primaryClass = "hep-ph",
    doi = "10.1007/JHEP08(2020)028",
    journal = "JHEP",
    volume = "08",
    number = "08",
    pages = "028",
    year = "2020"
}

@article{Larkoski:2017jix,
    author = "Larkoski, Andrew J. and Moult, Ian and Nachman, Benjamin",
    title = "{Jet Substructure at the Large Hadron Collider: A Review of Recent Advances in Theory and Machine Learning}",
    eprint = "1709.04464",
    archivePrefix = "arXiv",
    primaryClass = "hep-ph",
    doi = "10.1016/j.physrep.2019.11.001",
    journal = "Phys. Rept.",
    volume = "841",
    pages = "1--63",
    year = "2020"
}

@article{Thaler:2010tr,
    author = "Thaler, Jesse and Van Tilburg, Ken",
    title = "{Identifying Boosted Objects with N-subjettiness}",
    eprint = "1011.2268",
    archivePrefix = "arXiv",
    primaryClass = "hep-ph",
    reportNumber = "MIT-CTP-4191",
    doi = "10.1007/JHEP03(2011)015",
    journal = "JHEP",
    volume = "03",
    pages = "015",
    year = "2011"
}

@article{Larkoski:2014pca,
    author = "Larkoski, Andrew J. and Thaler, Jesse and Waalewijn, Wouter J.",
    title = "{Gaining (Mutual) Information about Quark/Gluon Discrimination}",
    eprint = "1408.3122",
    archivePrefix = "arXiv",
    primaryClass = "hep-ph",
    reportNumber = "MIT--CTP-4572, NIKHEF-2014-026",
    doi = "10.1007/JHEP11(2014)129",
    journal = "JHEP",
    volume = "11",
    pages = "129",
    year = "2014"
}

@book{Berge:1989,
  author = {Berge, Claude},
  title = {Hypergraphs: Combinatorics of Finite Sets},
  series = {North-Holland Mathematical Library},
  volume = {45},
  publisher = {North-Holland},
  year = {1989}
}

@article{Qu:2019gqs,
    author = "Qu, Huilin and Gouskos, Loukas",
    title = "{ParticleNet: Jet Tagging via Particle Clouds}",
    eprint = "1902.08570",
    archivePrefix = "arXiv",
    primaryClass = "hep-ph",
    doi = "10.1103/PhysRevD.101.056019",
    journal = "Phys. Rev. D",
    volume = "101",
    number = "5",
    pages = "056019",
    year = "2020"
}

@article{Amari:1998,
  author = {Amari, Shun-ichi},
  title = {Natural Gradient Works Efficiently in Learning},
  journal = {Neural Computation},
  volume = {10},
  number = {2},
  pages = {251--276},
  year = {1998},
  doi = {10.1162/089976698300017746}
}

@article{Komiske:2019fks,
    author = "Komiske, Patrick T. and Metodiev, Eric M. and Thaler, Jesse",
    title = "{Metric Space of Collider Events}",
    eprint = "1902.02346",
    archivePrefix = "arXiv",
    primaryClass = "hep-ph",
    reportNumber = "MIT-CTP 5102",
    doi = "10.1103/PhysRevLett.123.041801",
    journal = "Phys. Rev. Lett.",
    volume = "123",
    number = "4",
    pages = "041801",
    year = "2019"
}

@article{Nachman:2020ccu,
    author = "Nachman, Benjamin",
    title = "{Anomaly Detection for Physics Analysis and Less than Supervised Learning}",
    eprint = "2010.14554",
    archivePrefix = "arXiv",
    primaryClass = "hep-ph",
    month = "10",
    year = "2020"
}

@article{Blance:2019ibf,
  author = {Blance, Andrew and Spannowsky, Michael and Waite, Philip},
  title = {Adversarially-trained autoencoders for robust unsupervised new physics searches},
  journal = {JHEP},
  volume = {10},
  pages = {047},
  year = {2019},
  eprint = {1905.10384},
  archivePrefix = {arXiv},
  primaryClass = {hep-ph},
  reportNumber = {IPPP/19/41},
  doi = {10.1007/JHEP10(2019)047}
}

@article{Soper:2011cr,
  author = {Soper, Davison E. and Spannowsky, Michael},
  title = {Finding physics signals with shower deconstruction},
  journal = {Phys. Rev. D},
  volume = {84},
  pages = {074002},
  year = {2011},
  eprint = {1102.3480},
  archivePrefix = {arXiv},
  primaryClass = {hep-ph},
  doi = {10.1103/PhysRevD.84.074002}
}

@article{Konar:2021zdg,
  author = {Konar, Partha and Ngairangbam, Vishal S. and Spannowsky, Michael},
  title = {Energy-weighted message passing: an infra-red and collinear safe graph neural network algorithm},
  journal = {JHEP},
  volume = {02},
  pages = {060},
  year = {2022},
  eprint = {2109.14636},
  archivePrefix = {arXiv},
  primaryClass = {hep-ph},
  reportNumber = {IPPP/21/33},
  doi = {10.1007/JHEP02(2022)060}
}

@article{Alwall:2014hca,
    author         = "Alwall, J. and Frederix, R. and Frixione, S. and Hirschi, V.
                      and Maltoni, F. and Mattelaer, O. and Shao, H.-S. and Stelzer, T.
                      and Torrielli, P. and Zaro, M.",
    title          = "{The automated computation of tree-level and
                      next-to-leading order differential cross sections,
                      and their matching to parton shower simulations}",
    journal        = "JHEP",
    volume         = "07",
    year           = "2014",
    pages          = "079",
    doi            = "10.1007/JHEP07(2014)079",
    eprint         = "1405.0301",
    archivePrefix  = "arXiv",
    primaryClass   = "hep-ph"
}

@article{Ball:2017nwa,
    author         = "Ball, Richard D. and others",
    collaboration  = "NNPDF",
    title          = "{Parton distributions from high-precision collider data}",
    journal        = "Eur. Phys. J. C",
    volume         = "77",
    year           = "2017",
    number         = "10",
    pages          = "663",
    doi            = "10.1140/epjc/s10052-017-5199-5",
    eprint         = "1706.00428",
    archivePrefix  = "arXiv",
    primaryClass   = "hep-ph"
}

@article{Bertone:2017bme,
    author         = "Bertone, Valerio and Carrazza, Stefano and Hartland, Nathan P.
                      and Rojo, Juan",
    collaboration  = "NNPDF",
    title          = "{Illuminating the photon content of the proton within a
                      global PDF analysis}",
    journal        = "SciPost Phys.",
    volume         = "5",
    year           = "2018",
    number         = "1",
    pages          = "008",
    doi            = "10.21468/SciPostPhys.5.1.008",
    eprint         = "1712.07053",
    archivePrefix  = "arXiv",
    primaryClass   = "hep-ph"
}

@article{Sjostrand:2014zea,
    author         = "Sj{\"o}strand, Torbj{\"o}rn and Ask, Stefan and Christiansen,
                      Jesper R. and Corke, Richard and Desai, Nishita and Ilten, Philip
                      and Mrenna, Stephen and Prestel, Stefan and Rasmussen, Christine O.
                      and Skands, Peter Z.",
    title          = "{An introduction to \textsc{Pythia} 8.2}",
    journal        = "Comput. Phys. Commun.",
    volume         = "191",
    year           = "2015",
    pages          = "159--177",
    doi            = "10.1016/j.cpc.2015.01.024",
    eprint         = "1410.3012",
    archivePrefix  = "arXiv",
    primaryClass   = "hep-ph"
}

@article{CMS:2019csb,
    author         = "Sirunyan, Albert M and others",
    collaboration  = "CMS",
    title          = "{Extraction and validation of a new set of CMS
                      \textsc{Pythia8} tunes from underlying-event measurements}",
    journal        = "Eur. Phys. J. C",
    volume         = "80",
    year           = "2020",
    number         = "1",
    pages          = "4",
    doi            = "10.1140/epjc/s10052-019-7499-4",
    eprint         = "1903.12179",
    archivePrefix  = "arXiv",
    primaryClass   = "hep-ex"
}

@article{deFavereau:2013fsa,
    author         = "de Favereau, J. and Delaere, C. and Demin, P. and Giammanco, A.
                      and Lema{\^i}tre, V. and Mertens, A. and Selvaggi, M.",
    collaboration  = "DELPHES 3",
    title          = "{\textsc{Delphes} 3: a modular framework for fast simulation
                      of a generic collider experiment}",
    journal        = "JHEP",
    volume         = "02",
    year           = "2014",
    pages          = "057",
    doi            = "10.1007/JHEP02(2014)057",
    eprint         = "1307.6346",
    archivePrefix  = "arXiv",
    primaryClass   = "hep-ex"
}

@article{Butterworth:2008iy,
    author        = "Butterworth, Jonathan M. and Davison, Adam R. and Rubin, Mathieu and Salam, Gavin P.",
    title         = "{Jet substructure as a new Higgs search channel at the LHC}",
    eprint        = "0802.2470",
    archivePrefix = "arXiv",
    primaryClass  = "hep-ph",
    doi           = "10.1103/PhysRevLett.100.242001",
    journal       = "Phys. Rev. Lett.",
    volume        = "100",
    pages         = "242001",
    year          = "2008"
}

@article{Kaplan:2008ie,
    author        = "Kaplan, David E. and Rehermann, Keith and Schwartz, Matthew D. and Tweedie, Brock",
    title         = "{Top Tagging: A Method for Identifying Boosted Hadronically Decaying Top Quarks}",
    eprint        = "0806.0848",
    archivePrefix = "arXiv",
    primaryClass  = "hep-ph",
    doi           = "10.1103/PhysRevLett.101.142001",
    journal       = "Phys. Rev. Lett.",
    volume        = "101",
    pages         = "142001",
    year          = "2008"
}

@article{Kogler:2018hem,
    author        = "Kogler, Roman and others",
    title         = "{Jet Substructure at the Large Hadron Collider: Experimental Review}",
    eprint        = "1803.06991",
    archivePrefix = "arXiv",
    primaryClass  = "hep-ex",
    doi           = "10.1103/RevModPhys.91.045003",
    journal       = "Rev. Mod. Phys.",
    volume        = "91",
    number        = "4",
    pages         = "045003",
    year          = "2019"
}

@article{Kasieczka:2019dbj,
    author        = "Kasieczka, Gregor and others",
    title         = "{The Machine Learning landscape of top taggers}",
    eprint        = "1902.09914",
    archivePrefix = "arXiv",
    primaryClass  = "hep-ph",
    doi           = "10.21468/SciPostPhys.7.1.014",
    journal       = "SciPost Phys.",
    volume        = "7",
    pages         = "014",
    year          = "2019"
}

@article{Cacciari:2008gp,
    author = "Cacciari, Matteo and Salam, Gavin P. and Soyez, Gregory",
    title = "{The anti-$k_t$ jet clustering algorithm}",
    eprint = "0802.1189",
    archivePrefix = "arXiv",
    primaryClass = "hep-ph",
    reportNumber = "LPTHE-07-03",
    doi = "10.1088/1126-6708/2008/04/063",
    journal = "JHEP",
    volume = "04",
    pages = "063",
    year = "2008"
}

\end{document}